%% file: main.tex
\documentclass[letter,11pt,final]{article}

\input{macros.tex}

\title{Streaming Submodular Matching Meets the Primal-Dual Method}
\author[1]{Roie Levin\thanks{Supported in part by Anupam Gupta's NSF awards CCF-1907820, CCF1955785, and CCF-2006953.}} 
\affil[1]{Carnegie Mellon University}
\author[2]{David Wajc\thanks{Supported in part by NSF grants CCF-1527110, CCF-1618280, CCF-1814603, CCF-1910588, NSF CAREER award CCF-1750808 and a Sloan Research Fellowship.}}
\affil[2]{Stanford University}
\date{}

\begin{document}
\pagenumbering{gobble}
\maketitle

\input{abstract}

\newpage
\pagenumbering{arabic}

\input{intro.tex}

\input{preliminaries.tex}

\input{algorithm.tex}

\input{monotone.tex}

\input{nonmonotone.tex}

\input{linear.tex}

\input{lowerbound.tex}

\input{conclusion.tex}

\appendix
\section*{Appendix}
\input{other-work}
\input{tight.tex}

\input{extraproofs.tex}

\bibliography{abb,refs,ultimate}
\bibliographystyle{acmsmall}

\end{document}

%% file: macros.tex
\usepackage{fullpage}

\usepackage{mathtools}
\usepackage{authblk}
\usepackage{amsmath, amssymb, amsthm, thmtools, amsfonts, bm}
\usepackage{mathrsfs}
\usepackage{dsfont}
\usepackage{comment}
\usepackage{cite}
\usepackage[numbers,sort]{natbib}
\usepackage{url}
\usepackage[usenames,dvipsnames]{xcolor}
\usepackage{algorithm}
\usepackage{algorithmicx}
\usepackage[noend]{algpseudocode}
\usepackage{epstopdf}
\usepackage[textsize=tiny,disable]{todonotes}
\usepackage[shortlabels]{enumitem}
\usepackage{mleftright}
\usepackage{nicefrac}
\usepackage[skins]{tcolorbox}
\usepackage{wrapfig}
\usepackage[bold,full]{complexity}
\usepackage{subcaption}

\usepackage{framed}
\usepackage[framemethod=tikz]{mdframed}
\usepackage[bottom]{footmisc}
\usepackage[font=small]{caption}
\usepackage{xspace}
\usepackage{nicefrac}

\renewcommand\R{\mathbb{R}}
\DeclareMathOperator*{\expectation}{\mathbb{E}}
\let\poly\relax
\DeclareMathOperator*{\poly}{poly}

\DeclareMathOperator*{\eps}{\epsilon}
\newcommand\opt{OPT\xspace}
\newcommand\alg{\textsc{Alg}\xspace}
\newcommand\msm{\textsc{MSM}\xspace}
\newcommand\mwm{\textsc{MWM}\xspace}
\newcommand\msbm{\textsc{MSbM}\xspace}
\newcommand\mwbm{\textsc{MWbM}\xspace}
\newcommand\mkc{\textsc{Max k-Coverage}\xspace}
\newcommand\setc{\textsc{Set Cover}\xspace}
\newcommand\scopt{K\xspace}
\newcommand\scS{\mathcal{S}}
\newcommand\scU{\mathcal{U}}
\newcommand\fmin{f_{\min}}
\newcommand\fmax{f_{\max}}

\renewcommand{\E}{\mathbb{E}}%
\newcommand{\expect}{\expectation\expectarg}
\DeclarePairedDelimiterX{\expectarg}[1]{[}{]}{%
	\ifnum\currentgrouptype=16 \else\begingroup\fi
	\activatebar#1
	\ifnum\currentgrouptype=16 \else\endgroup\fi
}

\newcommand{\expectover}[1]{\expectation_{#1}\expectarg}

\newcommand{\innermid}{\nonscript\;\delimsize\vert\nonscript\;}
\newcommand{\activatebar}{%
	\begingroup\lccode`\~=`\|
	\lowercase{\endgroup\let~}\innermid 
	\mathcode`|=\string"8000
}

\usepackage{comment}
\usepackage[colorlinks,citecolor=blue,linkcolor=BrickRed]{hyperref}
\usepackage[capitalise]{cleveref}

\theoremstyle{plain}
\newtheorem{thm}{Theorem}[section]

\newtheorem{lem}[thm]{Lemma}

\newtheorem{Def}[thm]{Definition}
\newtheorem{obs}[thm]{Observation}

\crefname{cla}{Claim}{Claims}
\crefname{lem}{Lemma}{Lemmas}

\newenvironment{wrapper}[1]
{
	\begin{center}
		\begin{minipage}{\linewidth}
			\begin{mdframed}[hidealllines=true, backgroundcolor=gray!20, leftmargin=0cm,innerleftmargin=0.5cm,innerrightmargin=0.5cm,innertopmargin=0.5cm,innerbottommargin=0.5cm,roundcorner=10pt]
				#1}
			{\end{mdframed}
		\end{minipage}
	\end{center}
}

\newlength{\continueindent}
\usepackage{etoolbox}
\makeatletter
\newcommand*{\ALG@customparshape}{\parshape 2 \leftmargin \linewidth \dimexpr\ALG@tlm+\continueindent\relax \dimexpr\linewidth+\leftmargin-\ALG@tlm-\continueindent\relax}
\apptocmd{\ALG@beginblock}{\ALG@customparshape}{}{\errmessage{failed to patch}}
\makeatother


\usepackage{etoolbox}
\usepackage{tikz}
\usetikzlibrary{tikzmark}
\usetikzlibrary{calc}

\errorcontextlines\maxdimen

\newcommand{\ALGtikzmarkcolor}{black}
\newcommand{\ALGtikzmarkextraindent}{4pt}
\newcommand{\ALGtikzmarkverticaloffsetstart}{-.5ex}
\newcommand{\ALGtikzmarkverticaloffsetend}{-.5ex}
\makeatletter
\newcounter{ALG@tikzmark@tempcnta}

\newcommand\ALG@tikzmark@start{%
	\global\let\ALG@tikzmark@last\ALG@tikzmark@starttext%
	\expandafter\edef\csname ALG@tikzmark@\theALG@nested\endcsname{\theALG@tikzmark@tempcnta}%
	\tikzmark{ALG@tikzmark@start@\csname ALG@tikzmark@\theALG@nested\endcsname}%
	\addtocounter{ALG@tikzmark@tempcnta}{1}%
}

\def\ALG@tikzmark@starttext{start}
\newcommand\ALG@tikzmark@end{%
	\ifx\ALG@tikzmark@last\ALG@tikzmark@starttext
	\else
	\tikzmark{ALG@tikzmark@end@\csname ALG@tikzmark@\theALG@nested\endcsname}%
	\tikz[overlay,remember picture] \draw[\ALGtikzmarkcolor] let \p{S}=($(pic cs:ALG@tikzmark@start@\csname ALG@tikzmark@\theALG@nested\endcsname)+(\ALGtikzmarkextraindent,\ALGtikzmarkverticaloffsetstart)$), \p{E}=($(pic cs:ALG@tikzmark@end@\csname ALG@tikzmark@\theALG@nested\endcsname)+(\ALGtikzmarkextraindent,\ALGtikzmarkverticaloffsetend)$) in (\x{S},\y{S})--(\x{S},\y{E});%
	\fi
	\gdef\ALG@tikzmark@last{end}%
}

\apptocmd{\ALG@beginblock}{\ALG@tikzmark@start}{}{\errmessage{failed to patch}}
\pretocmd{\ALG@endblock}{\ALG@tikzmark@end}{}{\errmessage{failed to patch}}
\makeatother

\algblock[with]{With}{EndWith}
\algblockdefx[With]{With}{EndWith}%
[1]{\textbf{with} #1 \textbf{do}}%
{}

\makeatletter
\ifthenelse{\equal{\ALG@noend}{t}}%
{\algtext*{EndWith}}
{}%
\makeatother

\newcommand{\Continue}{\textbf{continue}}


%% file: abstract.tex
\begin{abstract}
We study streaming submodular maximization subject to matching/$b$-matching constraints (\msm/\msbm), and present improved upper and lower bounds for these problems. On the upper bounds front, we give primal-dual algorithms achieving the following approximation ratios.
\begin{itemize}
	\item $3+2\sqrt{2}\approx 5.828$ for monotone \msm, improving the previous best ratio of $7.75$.
	\item $4+3\sqrt{2}\approx 7.464$ for non-monotone \msm, improving the previous best ratio of $9.899$.
	\item $3+\eps$ for maximum weight b-matching, improving the previous best ratio of $4+\eps$.
\end{itemize}
On the lower bounds front, we improve on the previous best lower bound of $\frac{e}{e-1}\approx 1.582$ for \msm, and show ETH-based lower bounds of $\approx 1.914$ for polytime monotone \msm streaming algorithms. 

Our most substantial contributions are our algorithmic techniques. We show that the (randomized) primal-dual method, which originated in the study of maximum \emph{weight} matching (\mwm), is also useful in the context of \msm. To our knowledge, this is the first use of primal-dual based analysis for streaming submodular optimization. We also show how to reinterpret previous algorithms for \msm in our framework; hence, we hope our work is a step towards unifying old and new techniques for streaming submodular maximization, and that it paves the way for further new results.
	
\end{abstract}

%% file: intro.tex
\section{Introduction}
    
    In this paper we study streaming maximum submodular matching (\msm) and related problems. That is, we study the problem of computing high-value matchings where the value is determined by a submodular function, i.e. a function which captures \emph{diminishing returns}.

    \smallskip 
    
    Submodular function maximization has a long history. For example, it has been known since the 70s that the greedy algorithm yields an $e/(e-1) \approx 1.582$ approximation for monotone submodular maximization subject to a cardinality constraint \cite{nemhauser1978analysis}. This is optimal among polytime algorithms with value oracle access \cite{nemhauser1978best}, or assuming standard complexity-theoretic conjectures \cite{feige1998threshold,dinur2014analytical,manurangsi2020tight}. The same problem for \emph{non-monotone} submodular functions is harder; it is hard to approximate to within a $2.037$ factor \cite{oveisgharan2011submodular}. Much work has been dedicated to improving the achievable approximation \cite{lee2010maximizing,vondrak2013symmetry,oveisgharan2011submodular,feldman2011unified,buchbinder2014submodular,ene2016constrained,buchbinder2019constrained, gupta2010constrained}; the best currently stands at $2.597$  \cite{buchbinder2019constrained}.

    \smallskip    
    
	Closer to our work is the study of submodular maximization subject to \emph{matching} constraints.
	For this problem, the greedy algorithm has long been known to be $3$-approximate for monotone functions \cite{fisher1978analysis}. 
	Improved approximations have since been obtained \cite{lee2010submodular,lee2009non,gupta2010constrained,feldman2011improved}, with the current best being $(2+\epsilon)$ and $(4+\epsilon)$ for monotone and non-monotone MSM respectively \cite{feldman2011improved}.
	The papers above studied rich families of constraints (e.g. matroid intersection, matchoids, exchange systems), some of which were motivated explicitly by matching constraints (see  \cite{feldman2011improved}). Beyond theoretical interest, the MSM problem also has great practical appeal, since many natural objectives exhibit diminishing returns behavior. Applications across different fields include: machine translation \cite{lin2011word}, Internet advertising \cite{dickerson2019balancing,korula2018online}, combinatorial auctions more broadly \cite{vondrak2008optimal,lehmann2006combinatorial,calinescu2011maximizing}, and any matching problem where the goal is a submodular notion of utility such as diversity \cite{agrawal2009diversifying, ahmed2017diverse}. 
	
	\smallskip 
	The proliferation of big-data applications such as those mentioned above has spurred a surge of interest in algorithms for the regime where the input is too large to even store in local memory. To this end, it is common to formulate problems in the \textit{streaming} model. Here the input is presented element-by-element to an algorithm that is restricted to use $\tilde{O}(S)$ memory, where $S$ is the maximum size of any feasible solution. We study \msm in this model.
   
   \smallskip 
    
    For our problem when the objective is \emph{linear}, a line of work \cite{feigenbaum2005graph,mcgregor2005finding,epstein2013improved,crouch2014improved,paz20182+,ghaffari2019simplified} has shown that a $\left(2+\eps\right)$-approximation is possible in the streaming model \cite{paz20182+,ghaffari2019simplified}.  
    Meanwhile, for submodular objectives under \emph{cardinality} constraints (which are a special case of MSM in complete bipartite graphs), a separate line of work \cite{badanidiyuru2014streaming, kazemi2019submodular,alaluf2020optimal, chekuri2015streaming,feldman2018less,mirzasoleiman2018streaming} has culminated in the same $\left(2+\eps\right)$ approximation ratio, for both monotone and non-monotone functions (the latter taking exponential time, as is to be expected from the lower bound of \cite{oveisgharan2011submodular}); moreover, this $(2+\epsilon)$ bound was recently proven to be tight \cite{feldman2020one,norouzi2018beyond,alaluf2020optimal}.
    On the other hand, for fully general MSM, the gap between known upper and lower bounds remain frustratingly large. \citet{chakrabarti2015submodular} gave a  $7.75$-approximate algorithm for \msm with monotone functions. For non-monotone functions,  \citet{chekuri2015streaming} gave a $\left(12+\eps\right)$-approximate algorithm,
    later improved by \citet{feldman2018less} to $5+2\sqrt{6}\approx 9.899$.
    The only known lower bound for monotone \msm is $\frac{e}{e-1}\approx 1.582$ for streaming or polytime algorithms, implied respectively by \cite{kapralov2013better} and \cite{feige1998threshold,nemhauser1978best}. For non-monotone functions, \cite{oveisgharan2011submodular} implies a hardness of $2.037$.
    Closing these gaps, especially from the algorithmic side, seems to require new ideas.
    
    \smallskip

    \subsection{Our Contributions}

    We present a number of improved results for streaming maximum submodular matching (\msm) and related problems.
    
    \smallskip
    
    Our first result is an improvement on the $7.75$ approximation of \cite{chakrabarti2015submodular} for monotone \msm.
\begin{wrapper}
    \begin{restatable}{thm}{MSMub}\label{MSM-ub-det}
    There exists a deterministic quasilinear-time streaming \msm algorithm for monotone functions which is $3+2\sqrt{2} \approx 5.828$ approximate.
    \end{restatable}
\end{wrapper}
	
	Our algorithm extends in various ways: First, it yields the same approximation ratio for submodular \emph{$b$-matchings}, where each node $v$ can be matched $b_v$ times, improving on the previous best $8$-approximations \cite{chekuri2015streaming,feldman2018less}. For the special case of linear functions (\mwm), our algorithm---with appropriate parameters---recovers the $\left(2+\eps\right)$-approximate algorithm of \cite{paz20182+}. For weighted $b$-matching (\mwbm), a slight modification of our algorithm yields a $\left(3+\eps\right)$-approximate algorithm, improving on the previous best $\left(4+\epsilon\right)$-approximation  \cite{crouch2014improved}.\footnote{Independently, Mohsen Ghaffari (private communication) obtained the same $(3+\epsilon)$-approximate for \mwbm.}
	
	\smallskip 
	
	Next, we improve on the $5+2\sqrt{6}\approx 9.899$ approximation of \cite{feldman2018less} for non-monotone \msm. 
\begin{wrapper}
    \begin{restatable}{thm}{MSMub}\label{MSM-ub-rand}
	There exists a randomized linear-time streaming \msm algorithm for non-monotone functions which is $4+2\sqrt{3}\approx 7.464$ approximate.
	\end{restatable}
\end{wrapper}

	\smallskip 

	Our non-monotone \msm algorithm's approximation ratio is better than the previous state-of-the-art  $7.75$-approximate \emph{monotone} \msm algorithm \cite{chakrabarti2015submodular}. Moreover, when applied to monotone functions, the algorithm of \Cref{MSM-ub-rand} yields the same approximation ratio as the deterministic algorithm of \Cref{MSM-ub-det}. 
	
	\smallskip
	
	We turn to proving hardness for monotone \msm.
    As stated before, the previous best lower bounds for this problem were $\frac{e}{e-1}\approx 1.582$.
      These lower bounds applied to either space-bounded \cite{kapralov2013better} or time-bounded algorithms \cite{nemhauser1978best, feige1998threshold}.
    We show that the problem becomes harder for algorithms which are both space bounded and time bounded. 
    This answers an open problem posed in the Bertinoro Workshop on Sublinear Algorithms 2014  \cite{MSMopenQ}, at least for time bounded algorithms.\footnote{We note briefly that such a bound does not follow from space lower bounds for cardinality constrained submodular maximization \cite{feldman2020one,norouzi2018beyond} (a special case of our setting, with a complete bipartite graph on $n$ and $k$ nodes), since a bound for that problem cannot be superlinear in $n$.}

\begin{wrapper}    
    \begin{restatable}{thm}{MSMLB}
    	\label{MSM-lb}
    No polytime streaming \msm algorithm for monotone functions is better than $1.914$ approximate.
    \end{restatable}
\end{wrapper}    
        
   Finally, to demonstrate that our techniques have the potential for wider applicability, we also use them to provide an alternative and unified proof of the results of \citet{chakrabarti2015submodular} and \citet{feldman2018less} for \msm, in \Cref{sec:CK+FKK}.
    
    \subsection{Our Techniques and Overview}\label{sec:msm-techniques}
    
    Our starting point is the  breakthrough result of \citet{paz20182+} for a special case of our problem---maximum weight matching (\mwm). They gave a $(2+\epsilon)$-approximate streaming algorithm by extending the local-ratio technique \cite{bar2001unified}. Subsequently, \citet{ghaffari2019simplified} simplified and slightly improved the  analysis of \cite{paz20182+}, by re-interpreting their algorithm in terms of the primal-dual method.\footnote{An equivalence between the local-ratio and primal-dual methods was established in \cite{bar2005equivalence}, though it does not capture the extension of the local ratio method in \cite{paz20182+}.}
    The primal-dual method is ubiquitous in the context of approximating linear objectives. 
    In this paper, we show that this method is also useful in the context of streaming submodular optimization, where to the best of our knowledge, it has not yet been used.
	For our primal-dual analysis, we rely on the \textit{concave-closure extension} for submodular functions which has a ``configuration LP''-like formulation.
    In particular, using this extension, we find that a natural generalization of the \mwm algorithm of \cite{paz20182+} (described in \Cref{sec:MSM-algo}) yields improved bounds for monotone \msm and its generalization to $b$-matchings. Our primal-dual analysis is robust in the sense that it allows for extensions and generalizations, as we now outline.
    
    \medskip
    \textbf{Our approach in a nutshell} (Sections \ref{sec:MSM-algo}+\ref{sec:monotone}). Our approach is to keep monotone dual solutions (initially zero), and whenever an edge arrives, 
    discard it if its dual constraint is already satisfied. 
    Edges whose dual constraint is not satisfied are added to a stack $S$, and relevant dual variables are increased, so as to satisfy their dual constraint. 
    Finally, we unwind the stack $S$, constructing a matching $M$ greedily.
    The intuition here is that the latter edge in the stack incident on a common edge have higher marginal gain than earlier such edges in the stack.
    More formally, we show that this matching $M$ has value at least some constant times the dual objective cost. 
    Weak LP duality and the choice of LP imply that $f(M)\geq \frac{1}{\alpha}\cdot f(S\cup OPT)$ for some $\alpha>1$, which implies our algorithm is $\alpha$-approximate for monotone \msm. 

    \medskip    
    \textbf{Extension 1} (\cref{sec:non-monotone}). Extending our approach, which gives $f(M)\geq \frac{1}{\alpha}\cdot f(S\cup OPT)$, to non-monotone functions $f$ seems challenging, since for such functions $f(S\cup OPT)$ can be arbitrarily smaller than $f(OPT)$. 
    To overcome this challenge, we note that our dual updates over-satisfy dual constraints of edges in $S$. We can therefore afford to randomly discard edges whose dual is not satisfied on arrival (and not add them to $S$), resulting in these edges' dual constraints holding \emph{in expectation}. This allows us to argue, via a generalization of the randomized primal-dual method of \citet{devanur2013randomized} (on which we elaborate in \Cref{sec:prelims}), that $\E[f(M)]\geq \frac{1}{\alpha}\cdot \E[f(S\cup OPT)]$. As $S$ contains each element with probability at most some $q$, a classic lemma of \cite{buchbinder2014submodular} allows us to show that $\E[f(S\cup OPT)]\geq (1-q)\cdot f(OPT)$, from which we get our results for non-monotone \msm.\footnote{Incidentally, for monotone functions, for which $E[f(M)]\geq \frac{1}{\alpha}\cdot E[f(S\cup OPT)]\geq \frac{1}{\alpha}\cdot f(OPT)$, this algorithm is $\alpha$ approximate. This is somewhat surprising, as this algorithm runs an $\alpha$-approximate monotone algorithm (and this analysis is tight, by \Cref{sec:tight-MSM}) on a random $q$-fraction of the input, suggesting an  $\alpha/q$ approximation. Nonetheless, we show that for $q$ not too small in terms of $\alpha$, we retain the same approximation ratio even after this sub-sampling.}
%
%
%
%
    Given the wide success of the randomized-primal dual method of \cite{devanur2013randomized} in recent years \cite{huang2018match,huang2018online,huang2019tight,fahrbach2020edge,huang2020online,huang2020fully,gupta2017online,tang2020towards,huang2020adwords}, we believe that our extension of this method in the context of submodular optimization will likely find other applications.
	
 \medskip    
	\textbf{Extension 2} (\cref{sec:linear-obj}). For maximum weight \emph{$b$-matching} (\mwbm), the dual updates when adding an edge to the stack are not high enough to satisfy this edge's dual constraint. 
	However, since we do cover each edge \emph{outside} the stack $S$, weak duality implies that a maximum-weight $b$-matching $M$ in the stack $S$ has value at least as high as $f(M)\geq \frac{1}{2+\eps}\cdot f(OPT\setminus S)$, and trivially at least as high as $f(M)\geq f(OPT\cap S)$. Combining these lower bounds on $f(M)$ imply our improved $(3+\epsilon)$ approximation ratio for  \mwbm. This general approach seems fairly general, and could find uses for other sub-additive objectives subject to downward-closed constraints.
	
	    \medskip	
	\textbf{Unifying Prior Work} (\Cref{sec:CK+FKK}). To demonstrate the usefulness of our primal-dual analysis, we also show that this (randomized) primal-dual approach gives an alternative,  unified way to analyze the \msm algorithms of \cite{chakrabarti2015submodular,feldman2018less}.

	\medskip    
	\textbf{Lower bound} (\cref{sec:MSM-lb}). Our lower bound instance makes use of two sources of hardness: computational hardness under ETH (\cite{feige1998threshold,dinur2014analytical}) and information-theoretic hardness resulting form the algorithm not knowing the contents or order of the stream in advance (\cite{goel2012communication}). 
	In particular, our proof embeds a submodular optimization problem (specifically, set cover) in parts of the linear instance of \cite{goel2012communication},
	and hence exploits the submodularity in the \msm objective. Interestingly, our lower bound of $1.914$ is higher than any convex combination of the previous hardness results we make use of, both of which imply a lower bound no higher than $e/(e-1)$.

%% file: preliminaries.tex
\section{Preliminaries}\label{sec:prelims}
A set function $f:2^N\to \mathbb{R}$ is \emph{submodular} if 
the marginal gains of adding elements to sets, denoted by $f_{S}(e) := f(S\cup \{e\}) - f(S)$, satisfy 
$f_S(e) \geq f_T(e)$ for $e\not\in T$ and $S\subseteq T\subseteq N.$
We say $f$ is \emph{monotone} if  $f(S)\leq f(T)$ for all $S\subseteq T\subseteq N$. Throughout this paper we assume oracle access to the submodular function.
The maximum submodular matching (MSM) problem is defined by a non-negative submodular function $f:2^E\to \mathbb{R}_{\geq 0}$, where $E$ is the edge-set of some $n$-node graph $G=(V,E)$, and feasible sets are matchings in $G$. The more general maximum submodular $b$-matching (\msbm) problem has as feasible sets subgraphs in which the degree of each vertex $v$ does not exceed $b_v$, for some input vector $\vec{b}$. 
Our objective is to design algorithms with low approximation ratio $\alpha\geq 1$, that is algorithms producing solutions $M$ such that $\expect*{f(M)}\geq \frac{1}{\alpha} \cdot f(OPT)$ for the smallest possible value of $\alpha$, where $OPT$ is an optimal solution.

For \emph{streaming} MSM, edges of $E$ are presented one at a time, and we are tasked with computing a matching in $G$ at the end of the stream, using little memory. In particular, we give algorithms using space within the output-sensitive bound of $\tilde{O}(M_{\max})$, where $M_{\max}$ is the maximum size of a $b$-matching (compared to the entire graph size, which  can be $\Omega(n^2)$). We note that for matchings, this is a finer requirement than the usual $\tilde{O}(n)$ space bound, since $M_{\max} = O(n)$. We note also that for monotone objectives, an optimal $b$-matching is maximal, and so $M_{\max} \leq 2|OPT|$. Consequently, such space is a trivial lower bound for outputting constant-approximate solutions.

On a technical note, we will only allow the algorithm to query the value oracle for $f$ on subsets currently stored in memory. We assume the range of $f$ is polynomially bounded (it is common to assume bounds on this quantity, e.g., \cite{wolsey1982analysis,gupta2020online}). More precisely, we assume that $\fmax / \fmin= n^{O(1)}$ (these max and min marginals are defined below). This implies in particular that we can store a value $f_{S}(e)$ using $O(\log n)$ bits.

\paragraph{Useful Notation} Throughout this paper we will rely on the following notation.
First, we denote by $e^{(1)},e^{(2)}, \dots, $ the edges in the stream, in order. For edges $e=e^{(i)},e'=e^{(j)}$, we write $e<e'$ if and only if $i<j$, i.e., if $e$ arrived before $e'$.
Similarly to 
\cite{chekuri2015streaming,feldman2018less}, we will also use $f(e:S):=f_{S\cap \{e'<e\}}(e)$ as shorthand for the marginal gain from adding $e$ to the set of elements which arrived before $e$ in $S$. One simple yet useful property of this notation is that $\sum_{e \in S} f(e:S) = f(S)$ (\cite[Lemma 1]{chekuri2015streaming}.) Other properties of this notation we will make use of, both easy consequences of submodularity, are $f(e:S)\leq f_S(e)$, as well as monotonicity of $f(e:S)$ in $S$, i.e., $f(e:A)\geq f(e:B)$ for $A\subseteq B$. Define the max and min marginals of $f$ to be
$\fmax := \displaystyle \max \{f(e) \mid e \in E \}$ and
$\fmin := \displaystyle \min \{f(e \mid S) \mid e \in E, S \subseteq E, \ f(e \mid S) \neq 0\}$.
Finally, for any edge $e$, we denote by $N(e) := \{e' \mid e'\cap e\neq \emptyset \}$ the set of neighboring edges of $e$. 

\subsection{The Primal-Dual Method in Our Setting}
As discussed in \Cref{sec:msm-techniques}, the main workhorse of our algorithms is the primal-dual method. In this method, we consider some linear program (LP) relaxation, and its dual LP. We then
design an algorithm which computes a (primal) solution of value $P$, and a feasible solution of value $D$, and show that $P\geq \frac{1}{\alpha}\cdot D$, which implies an approximation ratio of $\alpha$, by weak duality, since 
$$P\geq \frac{1}{\alpha}\cdot D \geq \frac{1}{\alpha}\cdot f(OPT).$$

\smallskip

For linear objectives, the first step of the primal-dual method---obtaining an LP relaxation---is often direct: write some integer linear program for the problem and drop the integrality constraints. For submodular objective functions, which are only naturally defined over vertices of the hypercube, $\vec{x} \in \{0,1\}^E$, and are not defined over fractional points $\vec{x}\in [0,1]^E\setminus \{0,1\}^E$, the first step of defining a relaxation usually requires extending $f$ to real vectors. For this, we use the \emph{concave closure} (see e.g.\cite{vondrak2007submodularity} for a survey of its history and further properties).
\begin{Def}
The \emph{concave closure} $f^+:[0,1]^E\to \mathbb{R}$ of a set function $f:2^E\to \mathbb{R}$ is given by
$$f^+(\vec{x}) := \max\left\{\sum_{T\subseteq E} \alpha_T \cdot f(T) \,\,\Bigg\vert\,\,  \sum_{T\subseteq E} \alpha_T=1, \, \alpha_T\geq 0\,\, \forall T\subseteq E,\, \sum_{T \ni e} \alpha_T = x_e \, \,\forall e\in E\right\}.$$
\end{Def}

In words, the concave closure is the maximum expected $f$-value of a random subset $T\subseteq E$, where the maximum is taken over all distributions matching the marginal probabilities given by $\vec{x}$. This is indeed an extension of set functions (and in particular submodular functions) to real-valued vectors, as this distribution must be deterministic for all $\vec{x}\in \{0,1\}^E$. Consequently, for any set $P\subseteq [0,1]^E$ containing the characteristic vector $\vec{x}_{OPT}$ of an optimal solution $OPT$, we have that $\max_{\vec{x}^*\in P} f^+(\vec{x}) \geq f^+(\vec{x}_{OPT}) = f(OPT)$.

\smallskip

Now, to define an LP relaxation for submodular maximization of some function $g$ subject to some linear constraints $A\vec{x}\leq \vec{c}$, we simply consider $\max\left\{g^+(\vec{x}) \,\,\big\vert\,\, A\vec{x}\leq \vec{c}\right\}$. For \msbm, we obtain the primal and dual programs given in \Cref{fig:MSbM-lp}.
\vspace{-0.3cm}
\begin{figure}[H]
\begin{align*}
	\begin{array}{rl|rl}
	\multicolumn{2}{c|}{\text{Primal } (P)} & \multicolumn{2}{c}{\text{Dual } (D)} \\ \hline  
	\max & \sum_{T \subseteq E} \alpha_T \cdot g(T) &
	\min & \mu + \sum_{v\in V}{b_v \cdot \phi_v} \\
	\text{subject to} & & \text{subject to} & \\
	\forall T\subseteq E: & \sum_{T \ni e} \alpha_T = x_e & \forall T \subseteq E: & \mu + \sum_{e \in T} \lambda_e\geq g(T) \\
	& \sum_{T \subseteq E} \alpha_T = 1 & \forall e \in E: & \sum_{v \in e} \phi_v \geq \lambda_e \\
	\forall v \in V: &  \sum_{e\ni v}x_{e} \leq b_v & &\\
	\forall e\in E, T\subseteq E: & x_{e}, \alpha_T \geq 0 & \forall v\in V: & \phi_v \geq 0
	\end{array} \label{LP_q}
\end{align*}
\vspace{-0.5cm}
	\caption{The LP relaxation of the \msbm problem and its dual}
	\label{fig:MSbM-lp}
	\vspace{-0.25cm}
\end{figure}


\subsection{Non-Monotone MSM: Extending the Randomized Primal-Dual Method}
To go from monotone to non-monotone function maximization, we make use of our dual updates resulting in dual solutions which over-satisfy (some) dual constraints. 
This allows us to randomly sub-sample edges with probability $q$ when deciding whether to insert them into $S$, 
and still have a dual solution which is feasible \emph{in expectation} over the choice of $S$. This is akin to the randomized primal-dual method of \citet{devanur2013randomized}, who introduced this technique in the context of maximum \emph{cardinality} and \emph{weighted} matching.
However, unlike in \cite{devanur2013randomized} (and subsequent work \cite{huang2018match,huang2018online,huang2019tight,fahrbach2020edge,huang2020online,huang2020fully,gupta2017online,tang2020towards,huang2020adwords}), for our problem the LP is not fixed. Specifically, we consider a different submodular function in our LP based on $S$, denoted by $g^S(T):=f(T\cup S)$. This results in \emph{random} primal and dual LPs, depending on the random set $S$. We show that our (randomized) dual solution is feasible for the obtained (randomized) dual LP in expectation over $S$.
Consequently, our expected solution's value is at least as high as some multiple of an expected solution to the dual LP, implying
\begin{equation}\label{expect-dual-feasibility-impication}
\mathbb{E}_S[f(M)]\geq \frac{1}{\alpha}\cdot \mathbb{E}_S[D] \geq \frac{1}{\alpha}\cdot \mathbb{E}_S [f(S\cup OPT)].
\end{equation}

Equation \eqref{expect-dual-feasibility-impication} retrieves our bound for monotone functions, for which $\mathbb{E}_S [f(S\cup OPT)]\geq f(OPT)$. To obtain bounds for non-monotone functions, we show that $\mathbb{E}_S[f(S\cup OPT)]\geq (1-q)\cdot f(OPT)$, by relying on the following lemma, due to \citet[Lemma 2.2]{buchbinder2014submodular}.

\smallskip

\begin{lem}[\cite{buchbinder2014submodular}] \label{technion-lemma} Let $h : 2^N \to R_{\geq 0}$ be a non-negative
	submodular function, and let $B$ be a random subset of $N$ containing every element of $N$ with probability at most $q$ (not necessarily independently), then
	$\expect*{h(B)} \geq (1-q) \cdot h(\emptyset)$.
\end{lem}

%% file: algorithm.tex
\section{Our Basic Algorithm}\label{sec:MSM-algo}

In this section we describe our monotone submodular $b$-matching algorithm, which we will use with slight modifications and different parameter choices in coming sections. 

The algorithm maintains a stack of edges $S$, initially empty, as well as  vertex potentials $\vec \phi\in \mathbb{R}^{|V|}$. 
To avoid storing zero-valued potentials (increasing space to $\tilde{O}(n)$, rather than $\tilde{O}(M_{\max})$), we implicitly represent these by keeping a mapping from vertices with non-zero potentials to their potentials, by using e.g., a hash table or balanced binary search tree. When an edge $e$ arrives, we compare the marginal value of this arriving edge with respect to the stack to the sum of vertex potentials of the edge's endpoints times a slack parameter $C$. If $C \cdot \sum_{v \in e} \phi_v$ is larger, we continue to the next edge. Otherwise, with probability $q$ we add the edge to the stack and increment the endpoint vertex potentials. At the end of the stream, we construct a $b$-matching greedily by unwinding the stack \textit{in reverse order}. The pseudocode is given in \cref{alg:MSbM}. 

\begin{algorithm}
	\caption{The MSbM Algorithm}
	\label{alg:MSbM}
	\begin{minipage}[t]{0.46\linewidth}
		\begin{algorithmic}[1]	
			\vspace{0.1cm}
			\Statex \underline{\textbf{Initialization}}
			\State $S\gets \text{emptystack}$
			\State $\forall v\in V:\,  \phi^{(0)}_v\gets 0$ (implicitly)
			\vspace{0.1cm}
			\Statex \underline{\textbf{Loop}}
			\For{$t \in \{1, \ldots, |E| \}$}
			\State $e \leftarrow e^{(t)}$ 
			\State $\forall v\in V:\, \phi^{(t)}_v\gets \phi^{(t-1)}_v$ 
			\If{$C \cdot \sum_{v\in e}\phi^{(t-1)}_{v} \geq f(e:S)$} \label{line:nm-cover-check}
			\State \Continue \Comment{skip edge $e$}
			\Else
			\With{probability $q$} \label{line:nm-sampling}
			\State $S.push(e)$
			\For{$v\in e$}
			\State $w_{ev} \leftarrow \frac{f(e:S) - \sum_{v\in e}\phi^{(t-1)}_{v}}{b_v}$
			\EndFor
			\For{$v\in e$}
			\State $\phi^{(t)}_{v} \gets \phi^{(t-1)}_{v} + w_{ev}$
			\EndFor
			\EndWith
			\EndIf
			\EndFor
		\end{algorithmic}
	\end{minipage}
	\hfill
	\begin{minipage}[t]{0.46\linewidth}
		\begin{algorithmic}[1]
			\setcounter{ALG@line}{14}
			\vspace{0.1cm}
			\Statex \underline{\textbf{Post-Processing}}
			\State $M \gets  \emptyset$  \label{line:unwind-MSbM-begin}
			\While{$S\neq \text{emptystack}$} 
			\State $e\gets S.pop()$
			\If{$|M \cap N(e)| < b_v$ for all $v\in e$} \label{line:nm-block-check}
			\State $M \gets M \cup \{e\}$
			\EndIf
			\EndWhile
			\State\Return $M$ \label{line:unwind-MSbM-end}			
		\end{algorithmic}
	\end{minipage}
	
	\hfill
\end{algorithm}

\Cref{alg:MSbM} clearly outputs a feasible $b$-matching.
In subsequent sections we analyze this algorithm for various values of the parameters $C$ and $q$. Before doing so, we note that this algorithm when run with $C=1+\Omega(1)$ is indeed a streaming algorithm, i.e., its space usage is as follows.

\begin{restatable}{lem}{spacebound}
	\label{space-bound}
		For any constant $\epsilon>0$, \Cref{alg:MSbM} run with $C=1+\epsilon$ uses $\tilde{O}(M_{\max})$ space.
\end{restatable}
The proof relies on each vertex $v$ only having $\tilde{O}(b_v)$ edges in the stack, since every edge incident on vertex $v$ inserted to the stack increases $\phi_v$ by a multiplicative factor of $(C-1)/b_v$, while the minimum and maximum non-zero values which $\phi_v$ can take are polynomially bounded in each other, due to $f$ being polynomially bounded. 
Since this bound holds in particular for vertices of a minimum vertex cover, the bound follows due to minimum vertex covers having size at most twice that of maximal $b$-matchings in the same graph.
See \Cref{sec:deferred-msm-algo} for the complete proof.

We further note that  \Cref{alg:MSbM} runs in time (quasi)linear in $|E|$, times evaluation time of $f$.
\begin{restatable}{lem}{algtime}\label{time-bound}
	A randomized (deterministic) implementation of \Cref{alg:MSbM} requires $O(1)$ ($O(\log n)$) operations and $O(1)$ function evaluations per arrival, followed by $\tilde{O}(M_{\max})$ time post-processing. Using $\tilde{O}(n)$ space, the deterministic time per arrival can be decreased to $O(1)$, by storing all $\phi_v$.
\end{restatable}

%% file: monotone.tex
\section{Monotone MSbM}\label{sec:monotone}

In this section we will consider a \emph{deterministic} instantiation of \Cref{alg:MSbM} (specifically, we will set $q=1$) in the context of monotone submodular $b$-matching.

To argue about the approximation ratio, we will fit a dual solution to this algorithm. Define the auxiliary submodular functions ${g^S: 2^{E} \rightarrow \R^+}$ to be $g^S(T) := f(S \cup T)$. We will work with the dual LP \hyperlink{LP_q}{(D)} for the function $g^S$, and consider the following dual solution.
\begin{align*}
\mu & := f(S) = g^S(\emptyset),\\ \phi_v & := C\cdot \phi^{(|E|)}_v \\
\lambda_e & := \begin{cases}
f(e:S) & e\not \in S \\
0 & e\in S.
\end{cases}
\end{align*}
We start by showing that the above is indeed dual feasible.

\begin{lem}\label{obs:m-g-dual-covered}
	The dual solution $(\mu, \vec \phi, \vec \lambda)$ is feasible for the LP $(D)$ with function $g^S$.
\end{lem}
\begin{proof}
	To see that the first set of constraints is satisfied, note that by submodularity of $f$
	\begin{align*}
	\sum_{e \in T} \lambda_e = \sum_{e \in T \setminus S} f(e:S) \geq \sum_{e \in T \setminus S} f_{S}(e) \geq f_{S}(T \setminus S) = f(S \cup T) - f(S) = g^S(T) - \mu.
	\end{align*}
	
	For the second set of constraints, note that an edge $e=e^{(t)}$ is not added to the stack if and only if the check at \cref{line:nm-cover-check} fails. Therefore, since $\phi^{(t)}_v$ values increase monotonically with $t$, we have
	\[  \sum_{v\in e} \phi_v = C\cdot \sum_{v\in e} \phi^{(|E|)}_v  \geq C\cdot \sum_{v\in e} \phi^{(t-1)}_v \geq  f(e:S) = \lambda_e. \qedhere\]	
\end{proof}

It remains to relate the value of the solution $M$ to the cost of this dual. We first prove an auxiliary relationship that will be useful:

\begin{lem}
	\label{lem:m_fgeqw}
	The $b$-matching $M$ output by \Cref{alg:MSbM} satisfies 
	$$f(M) \geq \frac{1}{2}\cdot \sum_{e \in S}\sum_{v \in e} b_v\cdot w_{ev}.$$
\end{lem}
\begin{proof}
	We first note that for any edge $e=e^{(t)}$ and $v\in e$, since $\phi^{(t-1)}_v = \sum_{e' \ni v, e' < e} w_e$, we have that
	\[f(e:S) = b_v\cdot w_{ev} + \sum_{u \in e} \phi^{(t-1)}_u \geq b_v\cdot w_{ev} + \phi^{(t-1)}_v  = b_v\cdot w_{ev}+\sum_{\substack{e'\ni v \\ e'<e }}w_{e'v}.\]
	Combined with submodularity of $f$, the above yields the following lower bound on $f(M)$,
	\begin{align*}
	f(M) & =
	\sum_{e \in M} f(e:M) \geq \sum_{e \in M} f(e:S) \geq \sum_{e\in M}\sum_{v \in e} \big(b_v\cdot w_{ev}+ \sum_{\substack{e'\ni v \\ e'<e}}w_{e'v}\big).
	\end{align*}
	On the other hand, the greedy manner in which we construct $M$ implies that any edge $e'\in S\setminus M$ must have at least one endpoint $v$ with $b_v$ edges $e>e'$ in $M$. Consequently, the term $w_{e'v}$ for such $e$ and $v$ is summed $b_v$ times in the above lower bound for $f(M)$. On the other hand, $b_v\cdot w_{ev} = b_u\cdot w_{eu}$ for $e=(u,v)$, by definition. From the above we obtain our desired inequality.
	\begin{align*}
	f(M) & \geq \sum_{e\in M}\sum_{v \in e} b_v\cdot w_{ev} + \frac{1}{2}\cdot \sum_{e\in S\setminus M} \sum_{v\in e} b_v\cdot w_{ev} \geq \frac{1}{2}\cdot \sum_{e\in S}\sum_{v\in e} b_v\cdot w_{ev}.\qedhere
	\end{align*}
\end{proof}

We can now bound the two terms in the dual objective separately with respect to the primal, using the following two corollaries of \Cref{lem:m_fgeqw}.

\begin{lem}
	\label{cor:m_fandphi}
	The $b$-matching $M$ output by \Cref{alg:MSbM} satisfies $f(M) \geq \frac{1}{2C} \sum_{v \in V} b_v\cdot \phi_v$.
\end{lem}

\begin{proof}
	Since $\phi_v = C\cdot \phi^{(|E|)}_v$, and $w_{ev} = \phi^{(t)}_v - \phi^{(t-1)}_v$ for all $v\in e = e^{(t)}$, \Cref{lem:m_fgeqw} implies that
	\[
	f(M) \geq \frac{1}{2}\cdot \sum_{e\in S}\sum_{v \in e} b_v\cdot w_{ev} = \frac{1}{2}\cdot \sum_{v \in V}\sum_{t=1}^{|E|} b_v\cdot \left(\phi^{(t)}_v - \phi^{(t-1)}_v\right) = \frac{1}{2}\cdot \sum_{v \in V} b_v\cdot \phi^{(|E|)}_v = \frac{1}{2C}\cdot \sum_{v\in V} b_v\cdot \phi_v.\qedhere\]
\end{proof}

\begin{lem}
	\label{cor:m_fandmu}
	The $b$-matching $M$ output by \Cref{alg:MSbM} satisfies $f(M) \geq  \left(1 - \frac{1}{C}\right) \mu$.
\end{lem}

\begin{proof}
	We note that $w_{e}>0$ for an edge $e = e^{(t)}$ if and only if $f(e:S) \geq C\cdot \sum_{v\in e} \phi^{(t-1)}_v$. Hence,
	\[b_v\cdot w_{ev} = f(e:S) - \sum_{v\in e} \phi^{(t-1)}_v  \geq \left(1-\frac{1}{C}\right)\cdot f(e:S).\]
	Combining the above with \cref{lem:m_fgeqw}, and again recalling that for $e=(u,v)$, we have that $b_v\cdot w_{ev} = b_u\cdot w_{eu}$, by definition, we obtain the desired inequality.
	\[f(M) \geq \frac{1}{2}\cdot \sum_{e \in S}\sum_{v \in e} b_v\cdot  w_{ev} \geq \left(1 - \frac{1}{C}\right) \sum_{e \in S} f(e:S) = \left(1 - \frac{1}{C}\right)  f(S). \qedhere \]
\end{proof}

Combining the above two corollaries and \Cref{obs:m-g-dual-covered} with LP duality, we can now analyze the algorithm's approximation ratio.
\begin{thm}
	\Cref{alg:MSbM} run with $q=1$ and $C$ on a monotone \msbm instance outputs a $b$-matching $M$ of value \[\left(2C+\frac{C}{C-1}\right) \cdot f(M) \geq f(\opt).\]
	This is optimized by taking $C=1+\frac{1}{\sqrt2}$, which yields a $3+2\sqrt{2}\approx 5.828$ approximation.
\end{thm}

\begin{proof}
	By weak LP duality and \Cref{obs:m-g-dual-covered}, together with monotonicity of $f$, we have that
	\[C \cdot\sum_v  b_v \cdot \phi_v + \mu \geq \max_T g^S(T) = \max_T f(S \cup T) \geq f(S \cup \opt) \geq f(\opt).\]
	Combining \cref{cor:m_fandphi} and \cref{cor:m_fandmu} and rearranging, we get the desired inequality,
	\[\left(2C + \frac{C}{C-1}\right) \cdot  f(M) \geq C \cdot\sum_v  b_v \cdot \phi_v + \mu \geq f(\opt). \qedhere\]
\end{proof}

In \Cref{sec:tight-MSM} we show that our analysis of \Cref{alg:MSbM} is tight.

We note that our analysis of this section required monotonicity, as we lower bounded $f(M)$ by (a multiple of) $f(S\cup OPT)\geq f(OPT)$, where the last step crucially relies on monotonicity. In the next section, we show how the use of randomness (namely, setting $q\neq 1$) allows us to obtain new results for \emph{non-monotone} \msm.

%% file: nonmonotone.tex
\section{Non-Monotone \msm}\label{sec:non-monotone}

In this section we consider \msm (so, $b_v=1$ for all $v$ in this section), for \emph{non-monotone} functions.

To extend our results to non-monotone \msm, we make use of the freedom to choose $q\not\in \{0,1\}$, resulting in a randomized algorithm. This will allow us to lower bound $\E_S[f(S\cup OPT)]$ in terms of $f(OPT)$. But first, we show that for appropriately chosen $q$, the output matching $M$ has high value compared to $\E_S[f(S\cup OPT)]$. The analysis of this fact will follow the same outline of \Cref{sec:monotone}, relying on LP duality, but with a twist.

For our dual fitting, we use the same dual solution as in \Cref{sec:monotone}.
However, this time this dual solution will only be feasible \emph{in expectation}, in the following sense. Since we now have $q \not \in \{0, 1\}$, \cref{alg:MSbM} is now a randomized algorithm, $S$ is a random set, $g^{S}$ is a random submodular function, and thus \hyperlink{LP_q}{(D)} is a random LP. Let $\expect{(D)}$ denote this LP, which is \hyperlink{LP_q}{(D)} with the submodular function $g(T):=\expectover{S}{g^{S}(T)}$. We now show that our dual solution's expectation is feasible for $\expect{(D)}$.

\begin{lem}\label{obs:nm-g-dual-covered}
	For $q \in [1/(2C + 1), 1/2]$, the expected dual solution $(\expect{\mu}, \expect{\vec \phi}, \expect{\vec \lambda})$ is feasible for the expected LP $\expect{(D)}$.
\end{lem}
\begin{proof}
	The first set of constraints is satisfied for any realization of the randomness. Indeed, as in the proof of \Cref{obs:m-g-dual-covered}, for any realization of $S$, by submodularity of $f$, we have
	\begin{align*}
	\sum_{e \in T} \lambda_e = \sum_{e \in T \setminus S} f(e:S) \geq \sum_{e \in T \setminus S} f_{S}(e) \geq f_{S}(T \setminus S) = f(S \cup T) - f(S) = g^S(T) - \mu.
	\end{align*}
	Consequently, taking expectation over $S$, we have that indeed, $\E_S[\mu] + \sum_{e\in T} \E_S[\lambda_e] \geq \E_S[g^S(T)]$. We now tun to proving the second set of constraints, which will only hold in expectation.
	
	Fix an edge $e=e^{(t)}$, and define the event $A_e := [f(e: S) \leq C\cdot \sum_{v \in V} \phi^{(t-1)}_v]$. 
	Then, by definition of $A_e$ and monotonicity of $\phi^{(t)}_v$ in $t$, we have that
	\begin{equation}\label{pre-covered}
	\expect*{\sum_{v \in e} \phi_v  | A_e} \geq 
	\expect*{C\cdot \sum_{v\in e} \phi^{(t-1)}_v | A_e} \geq \expect*{f(e: S) | A_e} = \expect*{\lambda_e | A_e}. 
	\end{equation}	
	We now prove the same inequality holds when conditioning on the complement, $\overline{A_e}$.
	
	Fix a realization of the randomness $R$ for which $\overline{A_e}$ holds.
	Then, $e= e^{(t)}$ fails the test in \cref{line:nm-cover-check}, and so with probability $q$, we have $\sum_{v\in e}\phi_v^{(t)} = \sum_{v\in e} (\phi_v^{(t-1)} + w_e) = 2\cdot f(e:S) - \sum_{v \in e} \phi_v^{(t-1)}$, and with probability $(1-q)$, we have $\sum_{v\in e} \phi_v^{(t)} = \sum_{v \in e} \phi_v^{(t-1)}$. Hence, in this case, as $q\leq \frac{1}{2}$, we have
	\[
	\expect*{\sum_{v \in e} \phi_v^{(t)} | R} = 2q \cdot f(e: S) + (1-2q)\cdot \sum_{v \in e} \phi_v^{(t-1)}\geq 2q \cdot f(e:S). \]
	Now, since $\phi_v \geq C\cdot \phi^{(t)}$, and $q\geq 1/(2C+1)$ and since
	$\lambda_e$ is set to $f(e: S)$ 
	if $e$ is not added to $S$ (with probability $1-q$) and set to zero
	otherwise, the above implies that
	\begin{align*}
	\expect*{\sum_{v \in e} \phi_v | R}
	&\geq	2q  C\cdot f(e: S) \geq (1-q) \cdot f(e: S) 	= \expect*{\lambda_e | R}.
	\end{align*}
	By the law of total expectation, taken over all $R\subseteq \overline{A_e}$, we have
	\begin{equation}\label{not-pre-covered}
	\expect*{\sum_{v \in e} \phi_v  | \overline{A_e}} \geq  \expect*{\lambda_e | \overline{A_e}}. 
	\end{equation}
	
	Combining inequalities \eqref{pre-covered} and \eqref{not-pre-covered} with the law of total expectation gives the desired inequality,
	\[\expect*{\sum_{v \in e} \phi_v } \geq \expect*{\lambda_e }. \qedhere \]
\end{proof}

To bound the performance of this section's randomized variant of \Cref{alg:MSbM}, we can reuse corollaries \ref{cor:m_fandphi} and \ref{cor:m_fandmu}, since these follow from \Cref{obs:m-g-dual-covered}, which holds for every realization of the random choices of the algorithm. We now use these corollaries, LP duality and \Cref{obs:nm-g-dual-covered}, together with \Cref{technion-lemma}, to analyze this algorithm.

\begin{thm}
	\label{thm:nonmonotone-mainthm}
	\Cref{alg:MSbM} run with $q=1/(2C+1)$ and $C$ on a non-monotone \msm instance outputs a matching $M$ of value \[\left(\frac{4C^2-1}{2C-2}\right) \cdot f(M) \geq f(\opt).\]
	This is optimized by taking $C=1+\frac{\sqrt{3}}{2}$, resulting in an approximation ratio of $4+2\sqrt{3}\approx 7.464$.
	Moreover, the same algorithm is $2C+C/(C-1)$ approximate for \emph{monotone} \msm.
\end{thm}
\begin{proof}
		First, by \cref{cor:m_fandphi} and \cref{cor:m_fandmu}, for every realization of the algorithm, we have 
		\begin{equation*}
		\left(2C + \frac{C}{C-1}\right) \cdot f(M) \geq \sum_v \phi_v + \mu,
		\end{equation*}
		and thus this relationship holds in expectation as well.
		\begin{equation}\label{expected-relationship-to-dual-cost}
		\left(2C + \frac{C}{C-1}\right) \cdot \E[f(M)] \geq \expect*{\sum_v \phi_v + \mu}.
		\end{equation}
		
		On the other hand, by \cref{obs:nm-g-dual-covered}, the expected dual LP solution is feasible for $\expect{(D)}$. Therefore, by weak LP duality, we have
		\begin{equation}\label{expected-weak-duality}
		\expect*{\sum_v \phi_v + \mu} \geq \max_T \expect*{g(T)} = \max_T \expect*{ f(S \cup T)} \geq \expect*{f(S \cup \opt)}.
		\end{equation}
		The result for monotone \msm follows from equations \eqref{expected-relationship-to-dual-cost} and \eqref{expected-weak-duality}, together with monotonicity implying $\expect*{f(S \cup \opt)}\geq f(\opt)$.
		
		For non-monotone \msm, let us define the additional auxiliary function $h: 2^{E} \rightarrow \R^+$, with $h(T) := f(\opt \cup T)$. Now note that by our sampling procedure, $S$ is a random subset of $E$ containing every edge with probability at most $q$. Hence,  by \cref{technion-lemma}, we have
		\begin{equation}\label{MSM-technion-lemma}
			\expect*{f(S \cup \opt)} = \expect{h(S)} \geq (1-q)\cdot  h(\emptyset) = (1-q)\cdot f(\opt).
		\end{equation}
		
		Combining equations \eqref{expected-relationship-to-dual-cost}, \eqref{expected-weak-duality} and \eqref{MSM-technion-lemma}, together with our choice of $q=1/(2C+1)$, the desired inequality follows by rearranging terms.
\end{proof}

Having explored the use of \Cref{alg:MSbM} for submodular matchings, we now turn to analyzing this algorithm in the context of streaming \emph{linear} objectives.

%% file: linear.tex
\section{Linear Objectives}
\label{sec:linear-obj}

In this section we address the use of \Cref{alg:MSbM} to matching and $b$-matching with linear objectives, i.e.,  \mwm and  \mwbm, using a deterministic variant, with $q=1$.

For  \mwm, this algorithm with $C=1+\epsilon$ is essentially the algorithm of \cite{paz20182+}, and so it retrieves the state-of-the-art $(2+\epsilon)$-approximation for this problem, previously analyzed in \cite{paz20182+,ghaffari2019simplified}. We therefore focus on  \mwbm, for which a simple modification of \Cref{alg:MSbM} yields a $3+\epsilon$ approximation, improving upon the previous best $4+\epsilon$ approximation due to \cite{crouch2014improved}.

The modification to \Cref{alg:MSbM} which we consider is a natural one: instead of computing $M$ greedily, we simply compute an optimal  \mwbm $M$ in the subgraph induced by $S$, using a polytime linear-space offline algorithm (e.g., \cite{anstee1987polynomial,gabow2018data}).
Trivially, the $b$-matching $M$ has weight at least 
\begin{equation}\label{weight-of-stack}
w(M) \geq w(OPT\cap S).
\end{equation}
Moreover, this $b$-matching has weight no lower than the greedily-constructed $b$-matching of lines \ref{line:unwind-MSbM-begin}-\ref{line:unwind-MSbM-end}. We use LP duality to show that this modified algorithm with $C=1+\epsilon$ outputs a $b$-matching $M$ of weight at least $w(M) \geq \frac{1}{2+\epsilon}\cdot w(OPT\setminus S)$.
\begin{lem}\label{weight-in-stack}
	Let $M$ be a  \mwbm in the stack $S$ obtained by running \Cref{alg:MSbM} with $C=1+\epsilon/2$ and $q=1$ until \Cref{line:unwind-MSbM-begin}. Then, we have $w(M)\geq \frac{1}{2+\eps} \cdot w(OPT\setminus S)$.
\end{lem}
\begin{proof}
	Consider the matching $M'$ obtained by greedily unwinding the stack, as in \Cref{alg:MSbM}. 
	Clearly, $w(M)\geq w(M')$.
	So, by \Cref{cor:m_fandphi}, we have $w(M)\geq \frac{1}{2+\eps}\cdot \sum_{v\in V} \phi_v$, for $\phi_v = C\cdot \phi^{(|E|)}_v$. To relate $\sum_{v \in V} \phi_v$ to $w(OPT)$, we show that the dual solution $(0,\vec \phi, \vec{w})$ is dual feasible for the  LP $(D)$ with function $w$.
	
	The first set of constraints are trivially satisfied, due to linearity of $w$, as 
	$0+\sum_{e \in T} w_e = w(T).$
	
	For the second set of constraints, note that an edge $e=e^{(t)}$ is not added to the stack if and only if the check at \cref{line:nm-cover-check} fails. Therefore, since $\phi^{(t)}_v$ values increase monotonically with $t$, we have
	\[  \sum_{v\in e} \phi_v = C\cdot \sum_{v\in e} \phi^{(|E|)}_v  \geq C\cdot \sum_{v\in e} \phi^{(t-1)}_v \geq  f(e:S) = w_e.\]	
	
	Therefore, by weak LP duality, we have $w(M) \geq 0 + \frac{1}{2+\eps}\cdot \sum_{v\in V} \phi_v \geq \frac{1}{2+\eps}\cdot w(OPT)$.
\end{proof}

We are now ready to analyze the approximation ratio of this  \mwbm algorithm.
\begin{thm}\label{mwbm-approx}
	For any $\epsilon\geq 0$, \Cref{alg:MSbM} run with $C=1+\epsilon/2$ and $q=1$ until 
	\Cref{line:unwind-MSbM-begin}, followed by 
	a linear-space offline \mwbm algorithm run on $S$ to compute a solution $M$ is a $(3+\epsilon)$-approximate streaming  \mwbm algorithm.
\end{thm}
\begin{proof}
	To see that this is a streaming algorithm, we recall that $|S|=\tilde{O}(\sum_v b_v)$, by \Cref{space-bound}. Since we compute $M$ by running an offline linear-space algorithm on the subgraph induced by $S$, therefore using $O(|S|)$ space for this last step, the desired space bound follows. 

	To analyze the algorithm's approximation ratio, 	
	let $\alpha\in [0,1]$ be the weighted fraction of $OPT$ in $S$. That is, $w(OPT\cap S) = \alpha\cdot w(OPT)$, and by linearity, $w(OPT\setminus S) = (1-\alpha)\cdot w(OPT)$. Therefore, by \Cref{weight-of-stack} and \Cref{weight-in-stack} we have the following.
	\begin{align*} w(M) & \geq w(OPT\cap S) = \alpha \cdot w(OPT).
	\\
	w(M) & \geq \frac{1}{2+\eps}\cdot w(OPT\setminus S) = \frac{1-\alpha}{2+\epsilon} \cdot w(OPT).
	\end{align*}
	We thus find that the approximation ratio of this algorithm is at most $1/\min\{\alpha,\frac{1-\alpha}{2+\epsilon}\} \leq 3+\epsilon$. 
\end{proof}

\textbf{Remark.} We note that this approach---dual covering constraints for elements outside of the algorithm's memory $S$, and solving the problem optimally for $S$---is rather general. In particular, it applies to matching under any sub-additive (not just submodular) set function $f$,
for which $f(OPT)\leq f(OPT\setminus S)+f(OPT\cap S)$. Moreover, this approach extends beyond matchings, to any downward-closed constraints, for which $OPT\setminus S$ and $OPT\cap S$ are both feasible solutions.
So, it seems like this approach could find applications to streaming algorithms for other objectives and constraints, provided dual feasibility can be guaranteed using a dual solution of value bounded by that of the output solution.

%% file: lowerbound.tex
\section{Lower Bound for \msm}\label{sec:MSM-lb}

Previous work shows that beating a $\frac{e}{e-1} \approx 1.582$ approximation for \msm in the streaming model is impossible for quasilinear space bounded algorithms \cite{kapralov2013better}, or polytime bounded algorithms \cite{feige1998threshold,dinur2014analytical,manurangsi2020tight}. In this section, we show that assuming the exponential time hypothesis (ETH), whereby $\NP \not \subseteq \TIME(2^{o(n)})$ \cite{impagliazzo2001complexity}, beating $1.914$ is impossible for any algorithm that is both space and time bounded. In particular, we will rely on seminal hardness of approximation results \setc from \cite{dinur2014analytical}. Recall:

\begin{Def}
	A \setc instance consists of a set system $(\scU, \scS)$, with $\scS\subseteq 2^{\scU}$. The goal is to pick the smallest number $k$ of  sets $S_1, \ldots, S_k \in \mathcal{S}$ such that $\left|\bigcup_{i \in [k]} S_i \right| = |\scU|$. We use $\scopt$ to denote the size of the minimal cover for the instance $(\scU, \scS)$, and $N = |\scU| + |\scS|$ to denote the description size.
\end{Def}

\begin{restatable}[Extension of Corollary 1.6 of \cite{dinur2014analytical}]{lem}{maxkcov}
	\label{lem:sc_hardness}
	Assuming ETH, every algorithm achieving an approximation ratio $(1 - \alpha) \ln |\scU|$ for \textsc{Set Cover} runs in time strictly greater than $2^{N^{\gamma \cdot \alpha}}$ for some $\gamma > 0$. Furthermore, this holds even under the assumptions that $|\scS| \leq K^{1/(\gamma \alpha)}$ and $|\scU| \leq |\scS|^{1/(\gamma \alpha)}$.
\end{restatable}
See \cref{sec:extraproofs} for a proof that the hardness holds even under the extra assumptions. To describe the instance, we will also use some extremal graph theory results from \cite{goel2012communication}.

\begin{Def}
	An $\alpha$-\textit{Ruzsa-Szemer\'{e}di graph} ($\alpha$-RS graph) is a bipartite graph $G = (P,Q,E)$ with $|P| = |Q| = n$ that is a union of induced matchings of size exactly $\alpha n$.
\end{Def}

\begin{thm}[Lemma 53 of \cite{goel2012communication}]
	\label{thm:RSgraphs}
	For any constant $\eps > 0$, there exists a family of  balanced bipartite $(1/2 - \eps)$-RS graphs with $n^{1+\Omega(1/\log \log n)}$ edges. 
\end{thm}

In what follows we will show a randomized reduction from \setc to streaming \msm. Specifically, we will show that if there is a polytime streaming algorithm for \msm achieving ratio better than $1.914$, then there is an algorithm for \setc violating \cref{lem:sc_hardness}. We proceed to describing our reduction.

\medskip

\textbf{The Reduction.} The input is a \setc instance $(\scU,\scS)$ for which the minimal cover contains $\scopt$ sets. Fix $n = 2^{k^{1/d}}$ for a degree $d$ to be determined later.

We create an underlying bipartite graph $G = (L, R, E)$ with $n \poly \log n$ vertices as follows. The left/right vertex sets are partitioned into $L = P \sqcup P'$, $R = Q \sqcup Q'$. We let $|P| = |Q| = 2n$, and we let $|P'| = |Q'| = n \cdot |\scS|/\scopt$.

The edge set $E$ arrives in two phases. In phase 1, all the edge of a set $E_1$ arrive, in phase 2 the edges of $E_2$ arrive. To define $E_1$, let $G_0$ be a fixed $(1/2 - \eps)$-RS graph with $m=\Omega \left(n^{1+1/\log \log n}\right)$ edges between $P$ and $Q$, and let this graph be the union of the matchings $M_1, \ldots, M_t$. Let $M_i'$ be a random subset of $M_i$ of size $(1/2 - \delta)n$ for a parameter $\eps < \delta < 1/4$ and let $E_1 = M_1' \cup \ldots \cup M_t'$. Choose one index $r \in [k]$ uniformly at random, and call $M'_r$ the \textit{distinguished matching}. Note that the index $r$ is unknown to the algorithm.

Define $E_2$ as follows. Let $P_1 \sqcup P_2 \sqcup \ldots \sqcup P_{n/\scopt}$ and $Q_1 \sqcup Q_2 \sqcup \ldots \sqcup Q_{n/\scopt}$ be partitions of the the vertices of $P$ and $Q$ respectively not matched by $M_r$ into subsets of size $\scopt$. Similarly, let $P'_1 \sqcup P'_2 \sqcup \ldots \sqcup P'_{n/\scopt}$ and $Q'_1 \sqcup Q'_2 \sqcup \ldots \sqcup Q'_{n/\scopt}$ be partitions of $P'$ and $Q'$ into subsets of size $|\scS|$. Let $F_i$ be the edges of the complete bipartite graph between $P_i$ and $Q'_i$, and let $G_i$ be the edges of the complete bipartite graph between $Q_i$ and $P'_i$. Finally, set $E_2 = \bigcup_i F_i \cup G_i$. See \Cref{fig:lb-picture}.

\begin{figure}[!h]
	\captionsetup[subfigure]{labelformat=empty}
	\centering
	\scalebox{0.75}{\input{LBtikz.tikz}}
	\caption{Illustration of lower bound instance.}
		\label{fig:lb-picture}	
	\subcaption{Red edges represent the edges of the distinguished matching $M_r$ in $E_1$, purple edges represent other edges in $E_1$, green edges represent edges of $E_2$. The red and purple edges together form the $(1/2-\epsilon)$-RS graph $G_0$, subsampled.}
\end{figure}
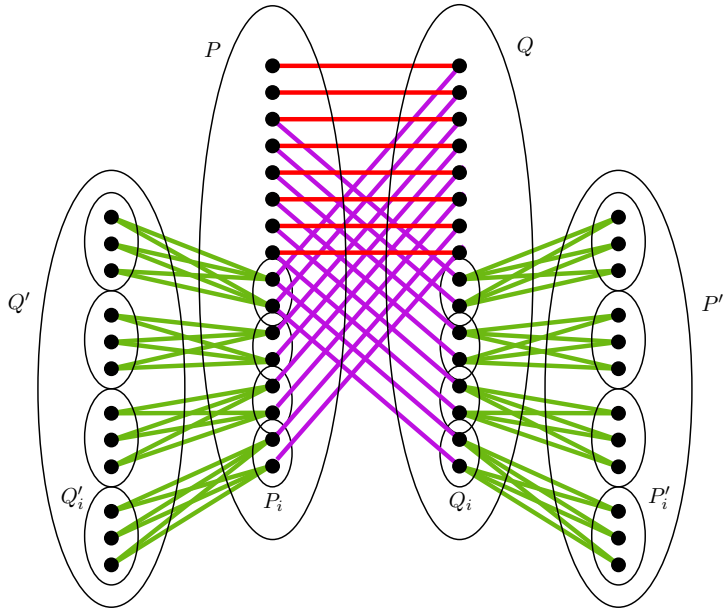

It remains to describe the submodular function $f$. First, define the set function $f_1(E) = |E \cap E_1|$. Next, we define the function $f_2$ which is parametrized by the \setc instance. We identify each set of vertices $P'_i$ and $Q'_i$ with a disjoint copy of $\scS$. For every edge $e\in E_2$, let $\phi(e)$ denote the set with which the endpoint of $e$ in $P' \cup Q'$ is associated. Now, for some parameter $\eta > 0$ to be determined later, we define
\[f_2(E) := \frac{\eta \scopt}{|\scU|} \cdot \left|\bigcup_{e \in E} \phi(e)\right|.\]
Finally, set $f := f_1 + f_2$. Note that $f$ is submodular since it the sum of a scaled coverage function and a linear function. On a technical note, since we assume that $|\scU|$ is polynomially bounded in $|\scS|$, we can represent the values of this function with $\poly \log n$ bits.

\medskip

\textbf{Some intuition.} Intuitively, we can imagine that all edges of $E_1$ are worth 1. We imagine that each edge of $E_2$ is a set in one of the copies of the instance $(\scS,\scU)$, and we let the value of all edges selected in the second phase be the coverage of all the associated sets (scaled by $\eta \scopt / |\scU|$). First we we will argue that the algorithm can output almost none of the edges of $M_r'$, since it after phase 1 it has no information as to which matching is the distinguished one. Hence the majority of the edges it uses from phase 1 must be from $E_1 \backslash M_r'$. However, each edge the algorithm chooses from $E_1 \backslash M_r'$ precludes it from taking between 1 and 2 edges of $E_2$. Furthermore, maximizing the value of edges of $E_2$ amounts to solving a hard $\mkc$ instance. The coverage is scaled by the parameter $\eta$, and as a result, the algorithm is incentivized to take some $k := c\scopt$ edges from each of the bipartite graphs $(P_i, Q_i')$ (and $(Q_i, P_i')$) of $E_2$ and the remaining edges from $E_1 \backslash M_r'$. Meanwhile, $\opt$ can take the distinguished matching edges $M_r'$ as well as the edges of $E_2$ maximizing the coverage instance. Our bound will follow by setting $\eta$ to maximize the ratio between these.

\medskip

To start, we show that no streaming algorithm can ``remember'' more than a $o(1)$ fraction of the edges of the distinguished matching $M_r$. Since phase 1 of our construction is identical to the one in Appendix H of \cite{goel2012communication} which shows a $3/2$ semi-streaming lower bound for max \textit{weight} matching., we can reuse their result here.
\begin{restatable}[Appendix H.1 of \cite{goel2012communication}]{lem}{gkk}
	\label{lem:gkk-forgetdistinguished}
	For any constants $\gamma, \delta  \in (0, 1/4)$, let $\mathcal{A}$ be an algorithm that at the end of phase 1, with constant probability, outputs at least $\gamma n$ of the the edges of $M_r'$. Then $\mathcal{A}$ uses $\Omega(E_1) \geq n^{1+\Omega(1/\log \log n)}$ bits of space.
\end{restatable}
We reproduce a version of the proof in \cref{sec:extraproofs} for completeness. With this, we are finally ready to prove the main theorem of the section.

\begin{thm}
	\label{thm:monotone-LB}
	Assuming ETH, there exists a distribution over \msm instances such that any deterministic algorithm achieving an $1.914$ approximation must use either $n^{1+\Omega(1/\log \log n)}$ space or $\Omega(2^{(\log n)^{10}})$ time.
\end{thm}

\begin{proof}
	Our proof is a randomized polytime reduction from \setc to streaming \msm. We will show that if there is a randomized streaming algorithm achieving ratio better than $1.914$ for \msm, then there is an algorithm for \setc achieving approximation ratio $(1-\alpha) \ln (|\scU|)$ for constant $\alpha > 0$ that only requires polynomial extra overheard. We then argue that \cref{lem:sc_hardness} implies that the streaming \msm algorithm must run in super polynomial time, assuming ETH.	
	
	Fix a deterministic algorithm $A$ for streaming \msm. Now, given an instance of \setc $(\scU, \scS)$ with minimum cover size $\scopt$ and description size $N = |\scU| + |\scS|$, create a random instance of streaming \msm according to the reduction described in this section. 
	For each bipartite graph $(P_i, Q'_i)$ (or $(Q_i, P'_i)$), if the algorithm $A$ chooses $c \scopt$ edges from this graph, it can select at most $2(1-c) \scopt$ edges from $E_1$ that are adjacent to $P_i$ (or $Q_i$). Suppose WLOG that it can always achieve the $2(1-c)k$ bound. In this case we can also assume WLOG the algorithm chooses the same number $c\cdot \scopt$ of edges from each such graph, and furthermore that it selects the same sets in the set system $(\scS, \scU)$. Otherwise it can locally improve its solution by copying the solution for the best index $i$. 
	
	Suppose this solution achieves coverage of $(1 - e^{-c} + \gamma) \cdot |\scU|$. Since the matchings $M_1, \ldots, M_t$ are induced, and by \cref{lem:gkk-forgetdistinguished} w.h.p. the algorithm can only output $o(n)$ edges of $M_r'$ after phase 1, the algorithm can only select $o(n)$ edges \textit{not} incident to some $P_i$ or $Q_i$. Thus the total value achieved by the algorithm is at most:
	\begin{align*}
	&\left[ 2(1-e^{-c} + \gamma) \cdot \eta \cdot \scopt + 2(1-c) \cdot \scopt \right] \cdot \frac{n}{\scopt} + o(n) \\
	&\leq  2(1-e^{-c} + \gamma) \cdot \eta n + 2(1-c) \cdot n + o(n)\\
	&\leq (2\eta  -2 \ln (\eta)  + 2 \gamma)\cdot n + o(n),
	\end{align*}
	where the last step follows since the expression is maximized at $c = \ln \eta$. On the other hand, the optimal solution can select the distinguished matching edge $M'_r$, as well as $\scopt$ edges adjacent to each set $P_i$ corresponding to the minimum \setc solution. Thus the total value of $\opt$ is at least:
	\[(1-\delta + 2 \eta) \cdot n.\]
	Thus the ratio between the maximum value achievable by the algorithm and the optimal value is bounded by:
	\[\frac{1 + 2 \eta -\delta}{2\eta  - 2\ln (\eta) + 2 \gamma +  o(1)}. \]
	Finally, we set $\eta = 2.09$ and let $\delta \rightarrow 0$. If this ratio converges to a value strictly below $1.914$, then we can conclude that $\gamma = \Omega(1)$ and $\gamma > 0$.

	We have shown that $A$ can be used to pick $c\scopt$ sets with coverage $(1 - e^{-c} + \gamma) \cdot |\scU|$. To finish the proof, we now show that this can be used to recover an approximation algorithm $B$ for \setc. For convenience, set constant $\gamma'$ such that $(1 - e^{- c - \gamma'}) := (1 - e^{- c} + \gamma)$. Then, guess $\scopt$, and repeat algorithm $A$ recursively $\lceil  \ln |\scU| /(c+\gamma')\rceil \leq \ln |\scU| /(c+\gamma') + 1$ times, each time on the residual uncovered set system. Each call to $A$ covers $(1 - e^{-c - \gamma'})$ of the elements remaining, so after this number of iterations, the fraction of uncovered elements is less than $1/|\scU|$, i.e. all elements are covered. Since each iteration costs $c \cdot \scopt$, the total number of sets picked here is at most
	\[c\left(\frac{\ln |\scU|}{c+\gamma'} + 1 \right)  \cdot \scopt= \left( \frac{c}{c+\gamma'} + \frac{c}{\ln |\scU|}\right)\cdot \ln |\scU| \cdot \scopt.\]
	Defining the constant $\alpha = \gamma'/(c+\gamma') $, this is a $(1-\alpha - o(1)) \ln |\scU|$ approximation to the \setc instance. Furthermore, if $A$ runs in time $T$  then $B$ runs in time $\poly(N) \cdot T$ (where $N = |\scU| + |\scS|$). 
	
	To conclude, if $T < 2^{N^{\Delta}}$ for a constant $\Delta < \gamma \cdot \alpha$, then $B$ runs faster than $2^{N^{\gamma \alpha}}$, contradicting \cref{lem:sc_hardness}. Thus the algorithm $A$ must run in time at least $2^{N^{\gamma \cdot \alpha}} \geq 2^{|\scS|^{\gamma \cdot \alpha}} \geq 2^{(\log n)^{d \cdot \gamma \cdot \alpha}}$. Setting $d = 10/(\gamma\cdot \alpha)$, this running time is $2^{(\log n)^{10}}$, which is superpolynomial in $n$. 
\end{proof}

 \cref{MSM-lb} therefore follows from \Cref{thm:monotone-LB} and Yao's minimax principle \cite{yao1977lemma}.

%% file: LBtikz.tikz
\tikzset{every picture/.style={line width=0.75pt}} 

\begin{tikzpicture}[x=0.75pt,y=0.75pt,yscale=-1,xscale=1]
\centering
\draw [color={rgb, 255:red, 108; green, 185; blue, 18 }  ,draw opacity=1 ][line width=2.25]    (379.11,196.34) -- (486,154.49) ;

\draw [color={rgb, 255:red, 108; green, 185; blue, 18 }  ,draw opacity=1 ][line width=2.25]    (379.11,214.34) -- (486,172.49) ;

\draw [color={rgb, 255:red, 108; green, 185; blue, 18 }  ,draw opacity=1 ][line width=2.25]    (379.11,214.34) -- (486,190.49) ;

\draw [color={rgb, 255:red, 108; green, 185; blue, 18 }  ,draw opacity=1 ][line width=2.25]    (379.11,196.34) -- (486,172.49) ;

\draw [color={rgb, 255:red, 108; green, 185; blue, 18 }  ,draw opacity=1 ][line width=2.25]    (379.11,196.34) -- (486,190.49) ;

\draw [color={rgb, 255:red, 108; green, 185; blue, 18 }  ,draw opacity=1 ][line width=2.25]    (379.11,214.3) -- (486,154.49) ;

\draw [color={rgb, 255:red, 108; green, 185; blue, 18 }  ,draw opacity=1 ][line width=2.25]    (379.11,232.26) -- (486,220.49) ;

\draw [color={rgb, 255:red, 108; green, 185; blue, 18 }  ,draw opacity=1 ][line width=2.25]    (379.11,251.85) -- (486,238.49) ;

\draw [color={rgb, 255:red, 108; green, 185; blue, 18 }  ,draw opacity=1 ][line width=2.25]    (379.11,251.85) -- (486,256.49) ;

\draw [color={rgb, 255:red, 108; green, 185; blue, 18 }  ,draw opacity=1 ][line width=2.25]    (379.11,232.26) -- (486,238.49) ;

\draw [color={rgb, 255:red, 108; green, 185; blue, 18 }  ,draw opacity=1 ][line width=2.25]    (379.11,232.26) -- (486,256.49) ;

\draw [color={rgb, 255:red, 108; green, 185; blue, 18 }  ,draw opacity=1 ][line width=2.25]    (379.11,251.81) -- (486,220.49) ;

\draw [color={rgb, 255:red, 108; green, 185; blue, 18 }  ,draw opacity=1 ][line width=2.25]    (379.11,268.18) -- (486,286.49) ;

\draw [color={rgb, 255:red, 108; green, 185; blue, 18 }  ,draw opacity=1 ][line width=2.25]    (379.11,286.14) -- (486,304.49) ;

\draw [color={rgb, 255:red, 108; green, 185; blue, 18 }  ,draw opacity=1 ][line width=2.25]    (379.11,286.14) -- (486,322.49) ;

\draw [color={rgb, 255:red, 108; green, 185; blue, 18 }  ,draw opacity=1 ][line width=2.25]    (379.11,268.18) -- (486,304.49) ;

\draw [color={rgb, 255:red, 108; green, 185; blue, 18 }  ,draw opacity=1 ][line width=2.25]    (377.62,267.58) -- (486,322.49) ;

\draw [color={rgb, 255:red, 108; green, 185; blue, 18 }  ,draw opacity=1 ][line width=2.25]    (379.11,286.14) -- (486,286.49) ;

\draw [color={rgb, 255:red, 108; green, 185; blue, 18 }  ,draw opacity=1 ][line width=2.25]    (379.11,304.1) -- (486,352.49) ;

\draw [color={rgb, 255:red, 108; green, 185; blue, 18 }  ,draw opacity=1 ][line width=2.25]    (379.11,322.06) -- (486,370.49) ;

\draw [color={rgb, 255:red, 108; green, 185; blue, 18 }  ,draw opacity=1 ][line width=2.25]    (379.11,322.06) -- (486,388.49) ;

\draw [color={rgb, 255:red, 108; green, 185; blue, 18 }  ,draw opacity=1 ][line width=2.25]    (379.11,304.1) -- (486,370.49) ;

\draw [color={rgb, 255:red, 108; green, 185; blue, 18 }  ,draw opacity=1 ][line width=2.25]    (379.11,304.1) -- (486,388.49) ;

\draw [color={rgb, 255:red, 108; green, 185; blue, 18 }  ,draw opacity=1 ][line width=2.25]    (379.11,322.06) -- (486,352.49) ;

\draw [color={rgb, 255:red, 108; green, 185; blue, 18 }  ,draw opacity=1 ][line width=2.25]    (251.89,196.34) -- (145,154.49) ;

\draw [color={rgb, 255:red, 108; green, 185; blue, 18 }  ,draw opacity=1 ][line width=2.25]    (251.89,214.34) -- (145,172.49) ;

\draw [color={rgb, 255:red, 108; green, 185; blue, 18 }  ,draw opacity=1 ][line width=2.25]    (251.89,214.34) -- (145,190.49) ;

\draw [color={rgb, 255:red, 108; green, 185; blue, 18 }  ,draw opacity=1 ][line width=2.25]    (251.89,196.34) -- (145,172.49) ;

\draw [color={rgb, 255:red, 108; green, 185; blue, 18 }  ,draw opacity=1 ][line width=2.25]    (251.89,196.34) -- (145,190.49) ;

\draw [color={rgb, 255:red, 108; green, 185; blue, 18 }  ,draw opacity=1 ][line width=2.25]    (251.89,214.3) -- (145,154.49) ;

\draw [color={rgb, 255:red, 108; green, 185; blue, 18 }  ,draw opacity=1 ][line width=2.25]    (251.89,232.26) -- (145,220.49) ;

\draw [color={rgb, 255:red, 108; green, 185; blue, 18 }  ,draw opacity=1 ][line width=2.25]    (251.89,251.85) -- (145,238.49) ;

\draw [color={rgb, 255:red, 108; green, 185; blue, 18 }  ,draw opacity=1 ][line width=2.25]    (251.89,251.85) -- (145,256.49) ;

\draw [color={rgb, 255:red, 108; green, 185; blue, 18 }  ,draw opacity=1 ][line width=2.25]    (251.89,232.26) -- (145,238.49) ;

\draw [color={rgb, 255:red, 108; green, 185; blue, 18 }  ,draw opacity=1 ][line width=2.25]    (251.89,232.26) -- (145,256.49) ;

\draw [color={rgb, 255:red, 108; green, 185; blue, 18 }  ,draw opacity=1 ][line width=2.25]    (251.89,251.81) -- (145,220.49) ;

\draw [color={rgb, 255:red, 108; green, 185; blue, 18 }  ,draw opacity=1 ][line width=2.25]    (251.89,268.18) -- (145,286.49) ;

\draw [color={rgb, 255:red, 108; green, 185; blue, 18 }  ,draw opacity=1 ][line width=2.25]    (251.89,286.14) -- (145,304.49) ;

\draw [color={rgb, 255:red, 108; green, 185; blue, 18 }  ,draw opacity=1 ][line width=2.25]    (251.89,286.14) -- (145,322.49) ;

\draw [color={rgb, 255:red, 108; green, 185; blue, 18 }  ,draw opacity=1 ][line width=2.25]    (251.89,268.18) -- (145,304.49) ;

\draw [color={rgb, 255:red, 108; green, 185; blue, 18 }  ,draw opacity=1 ][line width=2.25]    (253.38,267.58) -- (145,322.49) ;

\draw [color={rgb, 255:red, 108; green, 185; blue, 18 }  ,draw opacity=1 ][line width=2.25]    (251.89,286.14) -- (145,286.49) ;

\draw [color={rgb, 255:red, 108; green, 185; blue, 18 }  ,draw opacity=1 ][line width=2.25]    (251.89,304.1) -- (145,352.49) ;

\draw [color={rgb, 255:red, 108; green, 185; blue, 18 }  ,draw opacity=1 ][line width=2.25]    (251.89,322.06) -- (145,370.49) ;

\draw [color={rgb, 255:red, 108; green, 185; blue, 18 }  ,draw opacity=1 ][line width=2.25]    (251.89,322.06) -- (145,388.49) ;

\draw [color={rgb, 255:red, 108; green, 185; blue, 18 }  ,draw opacity=1 ][line width=2.25]    (251.89,304.1) -- (145,370.49) ;

\draw [color={rgb, 255:red, 108; green, 185; blue, 18 }  ,draw opacity=1 ][line width=2.25]    (251.89,304.1) -- (145,388.49) ;

\draw [color={rgb, 255:red, 108; green, 185; blue, 18 }  ,draw opacity=1 ][line width=2.25]    (251.89,322.06) -- (145,352.49) ;

\draw   (240,313.31) .. controls (240,300.79) and (245.93,290.63) .. (253.25,290.63) .. controls (260.57,290.63) and (266.5,300.79) .. (266.5,313.31) .. controls (266.5,325.84) and (260.57,336) .. (253.25,336) .. controls (245.93,336) and (240,325.84) .. (240,313.31) -- cycle ;
\draw   (240,277.39) .. controls (240,264.87) and (245.93,254.71) .. (253.25,254.71) .. controls (260.57,254.71) and (266.5,264.87) .. (266.5,277.39) .. controls (266.5,289.92) and (260.57,300.08) .. (253.25,300.08) .. controls (245.93,300.08) and (240,289.92) .. (240,277.39) -- cycle ;
\draw   (240.14,241.31) .. controls (240.14,228.79) and (246.07,218.63) .. (253.39,218.63) .. controls (260.71,218.63) and (266.64,228.79) .. (266.64,241.31) .. controls (266.64,253.84) and (260.71,264) .. (253.39,264) .. controls (246.07,264) and (240.14,253.84) .. (240.14,241.31) -- cycle ;
\draw   (240.14,204.92) .. controls (240.14,192.39) and (246.07,182.24) .. (253.39,182.24) .. controls (260.71,182.24) and (266.64,192.39) .. (266.64,204.92) .. controls (266.64,217.45) and (260.71,227.61) .. (253.39,227.61) .. controls (246.07,227.61) and (240.14,217.45) .. (240.14,204.92) -- cycle ;
\draw   (366,313.31) .. controls (366,300.79) and (371.93,290.63) .. (379.25,290.63) .. controls (386.57,290.63) and (392.5,300.79) .. (392.5,313.31) .. controls (392.5,325.84) and (386.57,336) .. (379.25,336) .. controls (371.93,336) and (366,325.84) .. (366,313.31) -- cycle ;
\draw   (366,277.39) .. controls (366,264.87) and (371.93,254.71) .. (379.25,254.71) .. controls (386.57,254.71) and (392.5,264.87) .. (392.5,277.39) .. controls (392.5,289.92) and (386.57,300.08) .. (379.25,300.08) .. controls (371.93,300.08) and (366,289.92) .. (366,277.39) -- cycle ;
\draw   (366.14,241.31) .. controls (366.14,228.79) and (372.07,218.63) .. (379.39,218.63) .. controls (386.71,218.63) and (392.64,228.79) .. (392.64,241.31) .. controls (392.64,253.84) and (386.71,264) .. (379.39,264) .. controls (372.07,264) and (366.14,253.84) .. (366.14,241.31) -- cycle ;
\draw   (366.14,204.92) .. controls (366.14,192.39) and (372.07,182.24) .. (379.39,182.24) .. controls (386.71,182.24) and (392.64,192.39) .. (392.64,204.92) .. controls (392.64,217.45) and (386.71,227.61) .. (379.39,227.61) .. controls (372.07,227.61) and (366.14,217.45) .. (366.14,204.92) -- cycle ;
\draw [color={rgb, 255:red, 189; green, 16; blue, 224 }  ,draw opacity=1 ][line width=2.25]    (379.22,53.01) -- (253.5,196.69) ;

\draw [color={rgb, 255:red, 189; green, 16; blue, 224 }  ,draw opacity=1 ][line width=2.25]    (379.22,70.97) -- (253.5,214.65) ;

\draw [color={rgb, 255:red, 189; green, 16; blue, 224 }  ,draw opacity=1 ][line width=2.25]    (377.72,91.92) -- (252,235.6) ;

\draw [color={rgb, 255:red, 189; green, 16; blue, 224 }  ,draw opacity=1 ][line width=2.25]    (379.22,106.89) -- (253.5,250.57) ;

\draw [color={rgb, 255:red, 189; green, 16; blue, 224 }  ,draw opacity=1 ][line width=2.25]    (382.21,120.36) -- (256.49,264.04) ;

\draw [color={rgb, 255:red, 189; green, 16; blue, 224 }  ,draw opacity=1 ][line width=2.25]    (382.21,138.32) -- (256.49,282) ;

\draw [color={rgb, 255:red, 189; green, 16; blue, 224 }  ,draw opacity=1 ][line width=2.25]    (380.61,158.92) -- (254.89,302.6) ;

\draw [color={rgb, 255:red, 189; green, 16; blue, 224 }  ,draw opacity=1 ][line width=2.25]    (382.1,173.89) -- (256.38,317.57) ;

\draw [color={rgb, 255:red, 189; green, 16; blue, 224 }  ,draw opacity=1 ][line width=2.25]    (253.39,88.58) -- (379.11,196.34) ;

\draw [color={rgb, 255:red, 189; green, 16; blue, 224 }  ,draw opacity=1 ][line width=2.25]    (253.39,106.54) -- (379.11,214.3) ;

\draw [color={rgb, 255:red, 189; green, 16; blue, 224 }  ,draw opacity=1 ][line width=2.25]    (253.39,124.5) -- (379.11,232.26) ;

\draw [color={rgb, 255:red, 189; green, 16; blue, 224 }  ,draw opacity=1 ][line width=2.25]    (253.39,142.46) -- (379.11,250.22) ;

\draw [color={rgb, 255:red, 189; green, 16; blue, 224 }  ,draw opacity=1 ][line width=2.25]    (253.39,160.42) -- (379.11,268.18) ;

\draw [color={rgb, 255:red, 189; green, 16; blue, 224 }  ,draw opacity=1 ][line width=2.25]    (253.39,178.38) -- (379.11,286.14) ;

\draw [color={rgb, 255:red, 189; green, 16; blue, 224 }  ,draw opacity=1 ][line width=2.25]    (253.39,196.34) -- (379.11,304.1) ;

\draw [color={rgb, 255:red, 189; green, 16; blue, 224 }  ,draw opacity=1 ][line width=2.25]    (253.39,214.3) -- (379.11,322.06) ;

\draw [color={rgb, 255:red, 255; green, 0; blue, 0 }  ,draw opacity=1 ][line width=2.25]    (257.88,52.66) -- (379.11,52.66) ;

\draw  [fill={rgb, 255:red, 0; green, 0; blue, 0 }  ,fill opacity=1 ] (248.9,52.66) .. controls (248.9,50.18) and (250.91,48.17) .. (253.39,48.17) .. controls (255.87,48.17) and (257.88,50.18) .. (257.88,52.66) .. controls (257.88,55.14) and (255.87,57.15) .. (253.39,57.15) .. controls (250.91,57.15) and (248.9,55.14) .. (248.9,52.66) -- cycle ;
\draw  [fill={rgb, 255:red, 0; green, 0; blue, 0 }  ,fill opacity=1 ] (374.62,52.66) .. controls (374.62,50.18) and (376.63,48.17) .. (379.11,48.17) .. controls (381.59,48.17) and (383.6,50.18) .. (383.6,52.66) .. controls (383.6,55.14) and (381.59,57.15) .. (379.11,57.15) .. controls (376.63,57.15) and (374.62,55.14) .. (374.62,52.66) -- cycle ;
\draw [color={rgb, 255:red, 255; green, 0; blue, 0 }  ,draw opacity=1 ][line width=2.25]    (253.39,70.62) -- (379.11,70.62) ;

\draw [color={rgb, 255:red, 255; green, 0; blue, 0 }  ,draw opacity=1 ][line width=2.25]    (257.88,88.58) -- (379.11,88.58) ;

\draw [color={rgb, 255:red, 255; green, 0; blue, 0 }  ,draw opacity=1 ][line width=2.25]    (253.39,106.54) -- (379.11,106.54) ;

\draw [color={rgb, 255:red, 255; green, 0; blue, 0 }  ,draw opacity=1 ][line width=2.25]    (257.88,124.5) -- (379.11,124.5) ;

\draw [color={rgb, 255:red, 255; green, 0; blue, 0 }  ,draw opacity=1 ][line width=2.25]    (253.39,142.46) -- (379.11,142.46) ;

\draw [color={rgb, 255:red, 255; green, 0; blue, 0 }  ,draw opacity=1 ][line width=2.25]    (257.88,160.42) -- (379.11,160.42) ;

\draw [color={rgb, 255:red, 255; green, 0; blue, 0 }  ,draw opacity=1 ][line width=2.25]    (253.39,178.38) -- (379.11,178.38) ;

\draw  [fill={rgb, 255:red, 0; green, 0; blue, 0 }  ,fill opacity=1 ] (248.9,70.62) .. controls (248.9,68.14) and (250.91,66.13) .. (253.39,66.13) .. controls (255.87,66.13) and (257.88,68.14) .. (257.88,70.62) .. controls (257.88,73.1) and (255.87,75.11) .. (253.39,75.11) .. controls (250.91,75.11) and (248.9,73.1) .. (248.9,70.62) -- cycle ;
\draw  [fill={rgb, 255:red, 0; green, 0; blue, 0 }  ,fill opacity=1 ] (374.62,70.62) .. controls (374.62,68.14) and (376.63,66.13) .. (379.11,66.13) .. controls (381.59,66.13) and (383.6,68.14) .. (383.6,70.62) .. controls (383.6,73.1) and (381.59,75.11) .. (379.11,75.11) .. controls (376.63,75.11) and (374.62,73.1) .. (374.62,70.62) -- cycle ;
\draw  [fill={rgb, 255:red, 0; green, 0; blue, 0 }  ,fill opacity=1 ] (248.9,88.58) .. controls (248.9,86.1) and (250.91,84.09) .. (253.39,84.09) .. controls (255.87,84.09) and (257.88,86.1) .. (257.88,88.58) .. controls (257.88,91.06) and (255.87,93.07) .. (253.39,93.07) .. controls (250.91,93.07) and (248.9,91.06) .. (248.9,88.58) -- cycle ;
\draw  [fill={rgb, 255:red, 0; green, 0; blue, 0 }  ,fill opacity=1 ] (248.9,106.54) .. controls (248.9,104.06) and (250.91,102.05) .. (253.39,102.05) .. controls (255.87,102.05) and (257.88,104.06) .. (257.88,106.54) .. controls (257.88,109.02) and (255.87,111.03) .. (253.39,111.03) .. controls (250.91,111.03) and (248.9,109.02) .. (248.9,106.54) -- cycle ;
\draw  [fill={rgb, 255:red, 0; green, 0; blue, 0 }  ,fill opacity=1 ] (248.9,124.5) .. controls (248.9,122.02) and (250.91,120.01) .. (253.39,120.01) .. controls (255.87,120.01) and (257.88,122.02) .. (257.88,124.5) .. controls (257.88,126.98) and (255.87,128.99) .. (253.39,128.99) .. controls (250.91,128.99) and (248.9,126.98) .. (248.9,124.5) -- cycle ;
\draw  [fill={rgb, 255:red, 0; green, 0; blue, 0 }  ,fill opacity=1 ] (248.9,142.46) .. controls (248.9,139.98) and (250.91,137.97) .. (253.39,137.97) .. controls (255.87,137.97) and (257.88,139.98) .. (257.88,142.46) .. controls (257.88,144.94) and (255.87,146.95) .. (253.39,146.95) .. controls (250.91,146.95) and (248.9,144.94) .. (248.9,142.46) -- cycle ;
\draw  [fill={rgb, 255:red, 0; green, 0; blue, 0 }  ,fill opacity=1 ] (248.9,160.42) .. controls (248.9,157.94) and (250.91,155.93) .. (253.39,155.93) .. controls (255.87,155.93) and (257.88,157.94) .. (257.88,160.42) .. controls (257.88,162.9) and (255.87,164.91) .. (253.39,164.91) .. controls (250.91,164.91) and (248.9,162.9) .. (248.9,160.42) -- cycle ;
\draw  [fill={rgb, 255:red, 0; green, 0; blue, 0 }  ,fill opacity=1 ] (248.9,178.38) .. controls (248.9,175.9) and (250.91,173.89) .. (253.39,173.89) .. controls (255.87,173.89) and (257.88,175.9) .. (257.88,178.38) .. controls (257.88,180.86) and (255.87,182.87) .. (253.39,182.87) .. controls (250.91,182.87) and (248.9,180.86) .. (248.9,178.38) -- cycle ;
\draw  [fill={rgb, 255:red, 0; green, 0; blue, 0 }  ,fill opacity=1 ] (374.62,88.58) .. controls (374.62,86.1) and (376.63,84.09) .. (379.11,84.09) .. controls (381.59,84.09) and (383.6,86.1) .. (383.6,88.58) .. controls (383.6,91.06) and (381.59,93.07) .. (379.11,93.07) .. controls (376.63,93.07) and (374.62,91.06) .. (374.62,88.58) -- cycle ;
\draw  [fill={rgb, 255:red, 0; green, 0; blue, 0 }  ,fill opacity=1 ] (374.62,106.54) .. controls (374.62,104.06) and (376.63,102.05) .. (379.11,102.05) .. controls (381.59,102.05) and (383.6,104.06) .. (383.6,106.54) .. controls (383.6,109.02) and (381.59,111.03) .. (379.11,111.03) .. controls (376.63,111.03) and (374.62,109.02) .. (374.62,106.54) -- cycle ;
\draw  [fill={rgb, 255:red, 0; green, 0; blue, 0 }  ,fill opacity=1 ] (374.62,124.5) .. controls (374.62,122.02) and (376.63,120.01) .. (379.11,120.01) .. controls (381.59,120.01) and (383.6,122.02) .. (383.6,124.5) .. controls (383.6,126.98) and (381.59,128.99) .. (379.11,128.99) .. controls (376.63,128.99) and (374.62,126.98) .. (374.62,124.5) -- cycle ;
\draw  [fill={rgb, 255:red, 0; green, 0; blue, 0 }  ,fill opacity=1 ] (374.62,142.46) .. controls (374.62,139.98) and (376.63,137.97) .. (379.11,137.97) .. controls (381.59,137.97) and (383.6,139.98) .. (383.6,142.46) .. controls (383.6,144.94) and (381.59,146.95) .. (379.11,146.95) .. controls (376.63,146.95) and (374.62,144.94) .. (374.62,142.46) -- cycle ;
\draw  [fill={rgb, 255:red, 0; green, 0; blue, 0 }  ,fill opacity=1 ] (374.62,160.42) .. controls (374.62,157.94) and (376.63,155.93) .. (379.11,155.93) .. controls (381.59,155.93) and (383.6,157.94) .. (383.6,160.42) .. controls (383.6,162.9) and (381.59,164.91) .. (379.11,164.91) .. controls (376.63,164.91) and (374.62,162.9) .. (374.62,160.42) -- cycle ;
\draw  [fill={rgb, 255:red, 0; green, 0; blue, 0 }  ,fill opacity=1 ] (374.62,178.38) .. controls (374.62,175.9) and (376.63,173.89) .. (379.11,173.89) .. controls (381.59,173.89) and (383.6,175.9) .. (383.6,178.38) .. controls (383.6,180.86) and (381.59,182.87) .. (379.11,182.87) .. controls (376.63,182.87) and (374.62,180.86) .. (374.62,178.38) -- cycle ;
\draw   (329.72,191.47) .. controls (329.72,92.08) and (351.83,11.5) .. (379.11,11.5) .. controls (406.39,11.5) and (428.5,92.08) .. (428.5,191.47) .. controls (428.5,290.87) and (406.39,371.45) .. (379.11,371.45) .. controls (351.83,371.45) and (329.72,290.87) .. (329.72,191.47) -- cycle ;
\draw  [fill={rgb, 255:red, 0; green, 0; blue, 0 }  ,fill opacity=1 ] (374.62,196.34) .. controls (374.62,193.86) and (376.63,191.85) .. (379.11,191.85) .. controls (381.59,191.85) and (383.6,193.86) .. (383.6,196.34) .. controls (383.6,198.82) and (381.59,200.83) .. (379.11,200.83) .. controls (376.63,200.83) and (374.62,198.82) .. (374.62,196.34) -- cycle ;
\draw  [fill={rgb, 255:red, 0; green, 0; blue, 0 }  ,fill opacity=1 ] (374.62,214.3) .. controls (374.62,211.82) and (376.63,209.81) .. (379.11,209.81) .. controls (381.59,209.81) and (383.6,211.82) .. (383.6,214.3) .. controls (383.6,216.78) and (381.59,218.79) .. (379.11,218.79) .. controls (376.63,218.79) and (374.62,216.78) .. (374.62,214.3) -- cycle ;
\draw  [fill={rgb, 255:red, 0; green, 0; blue, 0 }  ,fill opacity=1 ] (374.62,232.26) .. controls (374.62,229.78) and (376.63,227.77) .. (379.11,227.77) .. controls (381.59,227.77) and (383.6,229.78) .. (383.6,232.26) .. controls (383.6,234.74) and (381.59,236.75) .. (379.11,236.75) .. controls (376.63,236.75) and (374.62,234.74) .. (374.62,232.26) -- cycle ;
\draw  [fill={rgb, 255:red, 0; green, 0; blue, 0 }  ,fill opacity=1 ] (374.62,250.22) .. controls (374.62,247.74) and (376.63,245.73) .. (379.11,245.73) .. controls (381.59,245.73) and (383.6,247.74) .. (383.6,250.22) .. controls (383.6,252.7) and (381.59,254.71) .. (379.11,254.71) .. controls (376.63,254.71) and (374.62,252.7) .. (374.62,250.22) -- cycle ;
\draw  [fill={rgb, 255:red, 0; green, 0; blue, 0 }  ,fill opacity=1 ] (374.62,268.18) .. controls (374.62,265.7) and (376.63,263.69) .. (379.11,263.69) .. controls (381.59,263.69) and (383.6,265.7) .. (383.6,268.18) .. controls (383.6,270.66) and (381.59,272.67) .. (379.11,272.67) .. controls (376.63,272.67) and (374.62,270.66) .. (374.62,268.18) -- cycle ;
\draw  [fill={rgb, 255:red, 0; green, 0; blue, 0 }  ,fill opacity=1 ] (374.62,286.14) .. controls (374.62,283.66) and (376.63,281.65) .. (379.11,281.65) .. controls (381.59,281.65) and (383.6,283.66) .. (383.6,286.14) .. controls (383.6,288.62) and (381.59,290.63) .. (379.11,290.63) .. controls (376.63,290.63) and (374.62,288.62) .. (374.62,286.14) -- cycle ;
\draw  [fill={rgb, 255:red, 0; green, 0; blue, 0 }  ,fill opacity=1 ] (374.62,304.1) .. controls (374.62,301.62) and (376.63,299.61) .. (379.11,299.61) .. controls (381.59,299.61) and (383.6,301.62) .. (383.6,304.1) .. controls (383.6,306.58) and (381.59,308.59) .. (379.11,308.59) .. controls (376.63,308.59) and (374.62,306.58) .. (374.62,304.1) -- cycle ;
\draw  [fill={rgb, 255:red, 0; green, 0; blue, 0 }  ,fill opacity=1 ] (374.62,322.06) .. controls (374.62,319.58) and (376.63,317.57) .. (379.11,317.57) .. controls (381.59,317.57) and (383.6,319.58) .. (383.6,322.06) .. controls (383.6,324.54) and (381.59,326.55) .. (379.11,326.55) .. controls (376.63,326.55) and (374.62,324.54) .. (374.62,322.06) -- cycle ;
\draw   (204,191.85) .. controls (204,92.66) and (226.11,12.25) .. (253.39,12.25) .. controls (280.67,12.25) and (302.78,92.66) .. (302.78,191.85) .. controls (302.78,291.04) and (280.67,371.45) .. (253.39,371.45) .. controls (226.11,371.45) and (204,291.04) .. (204,191.85) -- cycle ;
\draw  [fill={rgb, 255:red, 0; green, 0; blue, 0 }  ,fill opacity=1 ] (248.9,196.34) .. controls (248.9,193.86) and (250.91,191.85) .. (253.39,191.85) .. controls (255.87,191.85) and (257.88,193.86) .. (257.88,196.34) .. controls (257.88,198.82) and (255.87,200.83) .. (253.39,200.83) .. controls (250.91,200.83) and (248.9,198.82) .. (248.9,196.34) -- cycle ;
\draw  [fill={rgb, 255:red, 0; green, 0; blue, 0 }  ,fill opacity=1 ] (248.9,214.3) .. controls (248.9,211.82) and (250.91,209.81) .. (253.39,209.81) .. controls (255.87,209.81) and (257.88,211.82) .. (257.88,214.3) .. controls (257.88,216.78) and (255.87,218.79) .. (253.39,218.79) .. controls (250.91,218.79) and (248.9,216.78) .. (248.9,214.3) -- cycle ;
\draw  [fill={rgb, 255:red, 0; green, 0; blue, 0 }  ,fill opacity=1 ] (248.9,232.26) .. controls (248.9,229.78) and (250.91,227.77) .. (253.39,227.77) .. controls (255.87,227.77) and (257.88,229.78) .. (257.88,232.26) .. controls (257.88,234.74) and (255.87,236.75) .. (253.39,236.75) .. controls (250.91,236.75) and (248.9,234.74) .. (248.9,232.26) -- cycle ;
\draw  [fill={rgb, 255:red, 0; green, 0; blue, 0 }  ,fill opacity=1 ] (248.9,250.22) .. controls (248.9,247.74) and (250.91,245.73) .. (253.39,245.73) .. controls (255.87,245.73) and (257.88,247.74) .. (257.88,250.22) .. controls (257.88,252.7) and (255.87,254.71) .. (253.39,254.71) .. controls (250.91,254.71) and (248.9,252.7) .. (248.9,250.22) -- cycle ;
\draw  [fill={rgb, 255:red, 0; green, 0; blue, 0 }  ,fill opacity=1 ] (248.9,268.18) .. controls (248.9,265.7) and (250.91,263.69) .. (253.39,263.69) .. controls (255.87,263.69) and (257.88,265.7) .. (257.88,268.18) .. controls (257.88,270.66) and (255.87,272.67) .. (253.39,272.67) .. controls (250.91,272.67) and (248.9,270.66) .. (248.9,268.18) -- cycle ;
\draw  [fill={rgb, 255:red, 0; green, 0; blue, 0 }  ,fill opacity=1 ] (248.9,286.14) .. controls (248.9,283.66) and (250.91,281.65) .. (253.39,281.65) .. controls (255.87,281.65) and (257.88,283.66) .. (257.88,286.14) .. controls (257.88,288.62) and (255.87,290.63) .. (253.39,290.63) .. controls (250.91,290.63) and (248.9,288.62) .. (248.9,286.14) -- cycle ;
\draw  [fill={rgb, 255:red, 0; green, 0; blue, 0 }  ,fill opacity=1 ] (248.9,304.1) .. controls (248.9,301.62) and (250.91,299.61) .. (253.39,299.61) .. controls (255.87,299.61) and (257.88,301.62) .. (257.88,304.1) .. controls (257.88,306.58) and (255.87,308.59) .. (253.39,308.59) .. controls (250.91,308.59) and (248.9,306.58) .. (248.9,304.1) -- cycle ;
\draw  [fill={rgb, 255:red, 0; green, 0; blue, 0 }  ,fill opacity=1 ] (248.9,322.06) .. controls (248.9,319.58) and (250.91,317.57) .. (253.39,317.57) .. controls (255.87,317.57) and (257.88,319.58) .. (257.88,322.06) .. controls (257.88,324.54) and (255.87,326.55) .. (253.39,326.55) .. controls (250.91,326.55) and (248.9,324.54) .. (248.9,322.06) -- cycle ;
\draw  [fill={rgb, 255:red, 0; green, 0; blue, 0 }  ,fill opacity=1 ] (481.51,154.49) .. controls (481.51,152.01) and (483.52,150) .. (486,150) .. controls (488.48,150) and (490.49,152.01) .. (490.49,154.49) .. controls (490.49,156.97) and (488.48,158.98) .. (486,158.98) .. controls (483.52,158.98) and (481.51,156.97) .. (481.51,154.49) -- cycle ;
\draw  [fill={rgb, 255:red, 0; green, 0; blue, 0 }  ,fill opacity=1 ] (481.51,172.49) .. controls (481.51,170.01) and (483.52,168) .. (486,168) .. controls (488.48,168) and (490.49,170.01) .. (490.49,172.49) .. controls (490.49,174.97) and (488.48,176.98) .. (486,176.98) .. controls (483.52,176.98) and (481.51,174.97) .. (481.51,172.49) -- cycle ;
\draw  [fill={rgb, 255:red, 0; green, 0; blue, 0 }  ,fill opacity=1 ] (481.51,190.49) .. controls (481.51,188.01) and (483.52,186) .. (486,186) .. controls (488.48,186) and (490.49,188.01) .. (490.49,190.49) .. controls (490.49,192.97) and (488.48,194.98) .. (486,194.98) .. controls (483.52,194.98) and (481.51,192.97) .. (481.51,190.49) -- cycle ;
\draw   (468,171) .. controls (468,152.77) and (476.06,138) .. (486,138) .. controls (495.94,138) and (504,152.77) .. (504,171) .. controls (504,189.23) and (495.94,204) .. (486,204) .. controls (476.06,204) and (468,189.23) .. (468,171) -- cycle ;
\draw  [fill={rgb, 255:red, 0; green, 0; blue, 0 }  ,fill opacity=1 ] (481.51,220.49) .. controls (481.51,218.01) and (483.52,216) .. (486,216) .. controls (488.48,216) and (490.49,218.01) .. (490.49,220.49) .. controls (490.49,222.97) and (488.48,224.98) .. (486,224.98) .. controls (483.52,224.98) and (481.51,222.97) .. (481.51,220.49) -- cycle ;
\draw  [fill={rgb, 255:red, 0; green, 0; blue, 0 }  ,fill opacity=1 ] (481.51,238.49) .. controls (481.51,236.01) and (483.52,234) .. (486,234) .. controls (488.48,234) and (490.49,236.01) .. (490.49,238.49) .. controls (490.49,240.97) and (488.48,242.98) .. (486,242.98) .. controls (483.52,242.98) and (481.51,240.97) .. (481.51,238.49) -- cycle ;
\draw  [fill={rgb, 255:red, 0; green, 0; blue, 0 }  ,fill opacity=1 ] (481.51,256.49) .. controls (481.51,254.01) and (483.52,252) .. (486,252) .. controls (488.48,252) and (490.49,254.01) .. (490.49,256.49) .. controls (490.49,258.97) and (488.48,260.98) .. (486,260.98) .. controls (483.52,260.98) and (481.51,258.97) .. (481.51,256.49) -- cycle ;
\draw   (468,237) .. controls (468,218.77) and (476.06,204) .. (486,204) .. controls (495.94,204) and (504,218.77) .. (504,237) .. controls (504,255.23) and (495.94,270) .. (486,270) .. controls (476.06,270) and (468,255.23) .. (468,237) -- cycle ;
\draw  [fill={rgb, 255:red, 0; green, 0; blue, 0 }  ,fill opacity=1 ] (481.51,286.49) .. controls (481.51,284.01) and (483.52,282) .. (486,282) .. controls (488.48,282) and (490.49,284.01) .. (490.49,286.49) .. controls (490.49,288.97) and (488.48,290.98) .. (486,290.98) .. controls (483.52,290.98) and (481.51,288.97) .. (481.51,286.49) -- cycle ;
\draw  [fill={rgb, 255:red, 0; green, 0; blue, 0 }  ,fill opacity=1 ] (481.51,304.49) .. controls (481.51,302.01) and (483.52,300) .. (486,300) .. controls (488.48,300) and (490.49,302.01) .. (490.49,304.49) .. controls (490.49,306.97) and (488.48,308.98) .. (486,308.98) .. controls (483.52,308.98) and (481.51,306.97) .. (481.51,304.49) -- cycle ;
\draw  [fill={rgb, 255:red, 0; green, 0; blue, 0 }  ,fill opacity=1 ] (481.51,322.49) .. controls (481.51,320.01) and (483.52,318) .. (486,318) .. controls (488.48,318) and (490.49,320.01) .. (490.49,322.49) .. controls (490.49,324.97) and (488.48,326.98) .. (486,326.98) .. controls (483.52,326.98) and (481.51,324.97) .. (481.51,322.49) -- cycle ;
\draw   (468,303) .. controls (468,284.77) and (476.06,270) .. (486,270) .. controls (495.94,270) and (504,284.77) .. (504,303) .. controls (504,321.23) and (495.94,336) .. (486,336) .. controls (476.06,336) and (468,321.23) .. (468,303) -- cycle ;
\draw  [fill={rgb, 255:red, 0; green, 0; blue, 0 }  ,fill opacity=1 ] (481.51,352.49) .. controls (481.51,350.01) and (483.52,348) .. (486,348) .. controls (488.48,348) and (490.49,350.01) .. (490.49,352.49) .. controls (490.49,354.97) and (488.48,356.98) .. (486,356.98) .. controls (483.52,356.98) and (481.51,354.97) .. (481.51,352.49) -- cycle ;
\draw  [fill={rgb, 255:red, 0; green, 0; blue, 0 }  ,fill opacity=1 ] (481.51,370.49) .. controls (481.51,368.01) and (483.52,366) .. (486,366) .. controls (488.48,366) and (490.49,368.01) .. (490.49,370.49) .. controls (490.49,372.97) and (488.48,374.98) .. (486,374.98) .. controls (483.52,374.98) and (481.51,372.97) .. (481.51,370.49) -- cycle ;
\draw  [fill={rgb, 255:red, 0; green, 0; blue, 0 }  ,fill opacity=1 ] (481.51,388.49) .. controls (481.51,386.01) and (483.52,384) .. (486,384) .. controls (488.48,384) and (490.49,386.01) .. (490.49,388.49) .. controls (490.49,390.97) and (488.48,392.98) .. (486,392.98) .. controls (483.52,392.98) and (481.51,390.97) .. (481.51,388.49) -- cycle ;
\draw   (468,369) .. controls (468,350.77) and (476.06,336) .. (486,336) .. controls (495.94,336) and (504,350.77) .. (504,369) .. controls (504,387.23) and (495.94,402) .. (486,402) .. controls (476.06,402) and (468,387.23) .. (468,369) -- cycle ;
\draw  [fill={rgb, 255:red, 0; green, 0; blue, 0 }  ,fill opacity=1 ] (149.49,154.49) .. controls (149.49,152.01) and (147.48,150) .. (145,150) .. controls (142.52,150) and (140.51,152.01) .. (140.51,154.49) .. controls (140.51,156.97) and (142.52,158.98) .. (145,158.98) .. controls (147.48,158.98) and (149.49,156.97) .. (149.49,154.49) -- cycle ;
\draw  [fill={rgb, 255:red, 0; green, 0; blue, 0 }  ,fill opacity=1 ] (149.49,172.49) .. controls (149.49,170.01) and (147.48,168) .. (145,168) .. controls (142.52,168) and (140.51,170.01) .. (140.51,172.49) .. controls (140.51,174.97) and (142.52,176.98) .. (145,176.98) .. controls (147.48,176.98) and (149.49,174.97) .. (149.49,172.49) -- cycle ;
\draw  [fill={rgb, 255:red, 0; green, 0; blue, 0 }  ,fill opacity=1 ] (149.49,190.49) .. controls (149.49,188.01) and (147.48,186) .. (145,186) .. controls (142.52,186) and (140.51,188.01) .. (140.51,190.49) .. controls (140.51,192.97) and (142.52,194.98) .. (145,194.98) .. controls (147.48,194.98) and (149.49,192.97) .. (149.49,190.49) -- cycle ;
\draw   (163,171) .. controls (163,152.77) and (154.94,138) .. (145,138) .. controls (135.06,138) and (127,152.77) .. (127,171) .. controls (127,189.23) and (135.06,204) .. (145,204) .. controls (154.94,204) and (163,189.23) .. (163,171) -- cycle ;
\draw  [fill={rgb, 255:red, 0; green, 0; blue, 0 }  ,fill opacity=1 ] (149.49,220.49) .. controls (149.49,218.01) and (147.48,216) .. (145,216) .. controls (142.52,216) and (140.51,218.01) .. (140.51,220.49) .. controls (140.51,222.97) and (142.52,224.98) .. (145,224.98) .. controls (147.48,224.98) and (149.49,222.97) .. (149.49,220.49) -- cycle ;
\draw  [fill={rgb, 255:red, 0; green, 0; blue, 0 }  ,fill opacity=1 ] (149.49,238.49) .. controls (149.49,236.01) and (147.48,234) .. (145,234) .. controls (142.52,234) and (140.51,236.01) .. (140.51,238.49) .. controls (140.51,240.97) and (142.52,242.98) .. (145,242.98) .. controls (147.48,242.98) and (149.49,240.97) .. (149.49,238.49) -- cycle ;
\draw  [fill={rgb, 255:red, 0; green, 0; blue, 0 }  ,fill opacity=1 ] (149.49,256.49) .. controls (149.49,254.01) and (147.48,252) .. (145,252) .. controls (142.52,252) and (140.51,254.01) .. (140.51,256.49) .. controls (140.51,258.97) and (142.52,260.98) .. (145,260.98) .. controls (147.48,260.98) and (149.49,258.97) .. (149.49,256.49) -- cycle ;
\draw   (163,237) .. controls (163,218.77) and (154.94,204) .. (145,204) .. controls (135.06,204) and (127,218.77) .. (127,237) .. controls (127,255.23) and (135.06,270) .. (145,270) .. controls (154.94,270) and (163,255.23) .. (163,237) -- cycle ;
\draw  [fill={rgb, 255:red, 0; green, 0; blue, 0 }  ,fill opacity=1 ] (149.49,286.49) .. controls (149.49,284.01) and (147.48,282) .. (145,282) .. controls (142.52,282) and (140.51,284.01) .. (140.51,286.49) .. controls (140.51,288.97) and (142.52,290.98) .. (145,290.98) .. controls (147.48,290.98) and (149.49,288.97) .. (149.49,286.49) -- cycle ;
\draw  [fill={rgb, 255:red, 0; green, 0; blue, 0 }  ,fill opacity=1 ] (149.49,304.49) .. controls (149.49,302.01) and (147.48,300) .. (145,300) .. controls (142.52,300) and (140.51,302.01) .. (140.51,304.49) .. controls (140.51,306.97) and (142.52,308.98) .. (145,308.98) .. controls (147.48,308.98) and (149.49,306.97) .. (149.49,304.49) -- cycle ;
\draw  [fill={rgb, 255:red, 0; green, 0; blue, 0 }  ,fill opacity=1 ] (149.49,322.49) .. controls (149.49,320.01) and (147.48,318) .. (145,318) .. controls (142.52,318) and (140.51,320.01) .. (140.51,322.49) .. controls (140.51,324.97) and (142.52,326.98) .. (145,326.98) .. controls (147.48,326.98) and (149.49,324.97) .. (149.49,322.49) -- cycle ;
\draw   (163,303) .. controls (163,284.77) and (154.94,270) .. (145,270) .. controls (135.06,270) and (127,284.77) .. (127,303) .. controls (127,321.23) and (135.06,336) .. (145,336) .. controls (154.94,336) and (163,321.23) .. (163,303) -- cycle ;
\draw  [fill={rgb, 255:red, 0; green, 0; blue, 0 }  ,fill opacity=1 ] (149.49,352.49) .. controls (149.49,350.01) and (147.48,348) .. (145,348) .. controls (142.52,348) and (140.51,350.01) .. (140.51,352.49) .. controls (140.51,354.97) and (142.52,356.98) .. (145,356.98) .. controls (147.48,356.98) and (149.49,354.97) .. (149.49,352.49) -- cycle ;
\draw  [fill={rgb, 255:red, 0; green, 0; blue, 0 }  ,fill opacity=1 ] (149.49,370.49) .. controls (149.49,368.01) and (147.48,366) .. (145,366) .. controls (142.52,366) and (140.51,368.01) .. (140.51,370.49) .. controls (140.51,372.97) and (142.52,374.98) .. (145,374.98) .. controls (147.48,374.98) and (149.49,372.97) .. (149.49,370.49) -- cycle ;
\draw  [fill={rgb, 255:red, 0; green, 0; blue, 0 }  ,fill opacity=1 ] (149.49,388.49) .. controls (149.49,386.01) and (147.48,384) .. (145,384) .. controls (142.52,384) and (140.51,386.01) .. (140.51,388.49) .. controls (140.51,390.97) and (142.52,392.98) .. (145,392.98) .. controls (147.48,392.98) and (149.49,390.97) .. (149.49,388.49) -- cycle ;
\draw   (163,369) .. controls (163,350.77) and (154.94,336) .. (145,336) .. controls (135.06,336) and (127,350.77) .. (127,369) .. controls (127,387.23) and (135.06,402) .. (145,402) .. controls (154.94,402) and (163,387.23) .. (163,369) -- cycle ;
\draw   (436.61,270) .. controls (436.61,188.81) and (458.72,123) .. (486,123) .. controls (513.28,123) and (535.39,188.81) .. (535.39,270) .. controls (535.39,351.19) and (513.28,417) .. (486,417) .. controls (458.72,417) and (436.61,351.19) .. (436.61,270) -- cycle ;
\draw   (95.61,270) .. controls (95.61,188.81) and (117.72,123) .. (145,123) .. controls (172.28,123) and (194.39,188.81) .. (194.39,270) .. controls (194.39,351.19) and (172.28,417) .. (145,417) .. controls (117.72,417) and (95.61,351.19) .. (95.61,270) -- cycle ;

\draw (213,42) node    {$P$};
\draw (423.5,40) node    {$Q$};
\draw (380,346) node    {$Q_{i}$};
\draw (513.5,343) node    {$P'_{i}$};
\draw (254,346) node    {$P_{i}$};
\draw (119,343) node    {$Q'_{i}$};
\draw (549.5,212) node    {$P'$};
\draw (83,212) node    {$Q'$};

\end{tikzpicture}

%% file: conclusion.tex
\section{Conclusion and Open Questions}
In this work we present a number of improved bounds for streaming submodular (b-)matching. 
Our work suggests a number of natural follow-up directions.

The first is to tighten (and ideally close) the gap between upper and lower bounds for the problems studied in this work. More ideas may be needed in order to resolve the optimal approximation ratio for these problems.

On the techniques side, it would be interesting to see whether the primal-dual method can be used for further applications in streaming submodular algorithms, or to further unify the analysis of prior work. More broadly, it would be interesting to apply our extension of the randomized primal dual method of \cite{devanur2013randomized} to other problems. 

Finally, we point to a possible connection between our streaming matching algorithms, and \emph{dynamic} matching algorithms. \Cref{alg:MSbM} when applied to monotone $b$-matching problems stores in its stack a bounded-degree subgraph with $\tilde{O}(M_{\max}) = \tilde{O}(|OPT|)$ edges. Such size-optimal constant-approximate matching sparsifiers have proven particularly useful in the dynamic matching literature. See \cite{wajc2020rounding} for a discussion, and \cite{arar2018dynamic,gupta2013fully,peleg2016dynamic,bhattacharya2015deterministic,bernstein2015fully,bernstein2016faster} for more applications of such dynamic matching sparsifiers. Is there a fast dynamic algorithm which maintains a sparisifier similar to that of \Cref{alg:MSbM}? This would result in improved dynamic weighted matching algorithms, and possibly would result in the first dynamic \emph{submodular} matching algorithm. This would add to the recent interest in efficient dynamic algorithms for submodular optimization \cite{lattanzi2020fully,gupta2020fully}.

%% file: other-work.tex
\section{Explaining Prior Work using LP Duality}\label{sec:CK+FKK}

In this section we further demonstrate the generality of our (randomized) primal-dual analysis, showing that it provides fairly simple alternative analyses of the algorithms of \cite{chakrabarti2015submodular,feldman2018less}, giving one unified analysis for these algorithms and ours. To keep things simple, we focus only on \msm, though \cite{chakrabarti2015submodular,feldman2018less} show that their algorithms also work more broadly for $k$-matchoid and $k$-set system constraints.

In \cite{chakrabarti2015submodular}, Chakrabarti and Kale presented a reduction from \msm to \mwm, by showing how to use a subclass of \mwm algorithms to solve \msm. We now introduce the algorithm of \cite{chakrabarti2015submodular} instantiated with the \mwm algorithm of \citet{mcgregor2005finding}. The algorithm is a natural and elegant one: when an edge $e$ arrives, we consider its marginal gain with respect to the current matching. If this marginal gain is higher than some slack parameter $C$ times the marginal gains of the currently blocking edges $e'\in N(e)\cap M$, we preempt those edges and add $e$ to the matching. 

In anticipation of our analysis of the algorithm of \cite{feldman2018less} in \Cref{sec:FKK}, we generalize the algorithm's description and allow the algorithm to preempt with some probability $q \in [0,1]$. The full pseudo-code is given in \Cref{alg:MSbM-CK-FKK}.

\begin{tcolorbox}[blanker,float=htbp,
	grow to left by=-0.225\textwidth, 	grow to right by=-0.225\textwidth]
	\begin{algorithm}[H]
		\caption{The MSM Algorithm of \cite{chakrabarti2015submodular} and \cite{feldman2018less}}
		\label{alg:MSbM-CK-FKK}
		\begin{algorithmic}[1]	
			\vspace{0.1cm}
			\Statex \underline{\textbf{Initialization}}
			\State $M\gets \emptyset$
			\vspace{0.1cm}
			\Statex \underline{\textbf{Loop}}
			\For{$t \in \{1, \ldots, |E| \}$}
			\State $e \leftarrow e^{(t)}$ 
			\State $B(e) \leftarrow \sum_{e'\in N(e)\cap M}f(e':M)$
			\If{$f(e:M) \leq C \cdot B(e)$} \label{line:nm-cover-check-CK}
			\State \Continue \Comment{skip edge $e$}
			\Else
			\With{probability $q$} \label{line:nm-sampling-CK}
			\State $M\gets (M\setminus N(e))\cup \{e\}$
			\EndWith
			\EndIf
			\EndFor
			\State \textbf{return} $M$
		\end{algorithmic}
		\vspace{-0.4cm}
		\hfill
	\end{algorithm}
\end{tcolorbox}

This algorithm only ever adds an edge $e$ to $M$ upon its arrival. After adding an edge to $M$, this edge can be \emph{preempted}, i.e., removed from $M$, after which it is never added back to $M$. Thus we note that this algorithm is not only a streaming algorithm, but also a so-called \emph{preemptive} algorithm: it only stores a single matching in memory and therefore trivially requires $\tilde{O}(n)$ space. 

For convenience, we let $M^{(t)}$ denote the matching $M$ at time $t$, and let $S:=\bigcup_t M^{(t)}$ denote the set of edges ever added to $M$. For an edge $e$ let $B^{(t)}(e) := \sum_{e' \in N(e) \cap M^{(t)}} f(e' : M^{(t)})$. We will also denote by $P:=S\setminus M$ the set of \emph{preempted} edges.

\input{McGandCK.tex}

\input{FKK.tex}

%% file: McGandCK.tex
\subsection{The Framework of \cite{chakrabarti2015submodular}, Applied to the Algorithm of \cite{mcgregor2005finding}}

\label{sec:CK+mcG}

In this section we analyze the deterministic algorithm obtained by applying the framework of \citet{chakrabarti2015submodular} to the \mwm algorithm of \citet{mcgregor2005finding}, corresponding to \Cref{alg:MSbM-CK-FKK} run with $q=1$.

To argue about the approximation ratio, we will again fit a dual solution to this algorithm. Define the auxiliary submodular functions ${g^S: 2^{E} \rightarrow \R^+}$ to be $g^S(T) := f(S \cup T)$. 
Similarly to our analysis of \Cref{alg:MSbM}, we define the following dual.
\begin{align*}
	\mu & := f(S) = g^S(\emptyset), \\
	\phi_v & := C\cdot \max\{f(e:M^{(t)}) \mid  t \in [|E|], \ v \in e\in M^{(t)}\}, \\
	\lambda_e & := \begin{cases}
	f(e:S) & e\not\in S \\
	0 & e\in S.
	\end{cases}
\end{align*}
Note the difference here in the setting of $\phi_v$ from the algorithms of \cref{sec:monotone,sec:non-monotone}.
We start by showing that this is a dual feasible solution to the LP \hyperlink{LP_q}{(D)} for the function $g^S$. 

\begin{lem}\label{obs:m-g-dual-covered-CK}
	The dual solution $(\vec \lambda, \vec \phi, \mu)$ is feasible for the LP $(D)$ with function $g^S$.
\end{lem}
\begin{proof}
	To see that the first set of constraints are satisfied, note that by submodularity of $f$,
	\begin{align*}
	\sum_{e \in T} \lambda_e = \sum_{e \in T \setminus S} f(e:S) \geq \sum_{e \in T \setminus S} f_{S}(e) \geq f_{S}(T \setminus S) = f(S \cup T) - f(S) = g^S(T) - \mu.
	\end{align*}
	
	For the second set of constraints, we note that if $e=e^{(t)}\not\in S$, then by the test in \Cref{line:nm-cover-check-CK} and submodularity, we have that
	\[
	\lambda_e = f(e:S) \leq f(e:M^{(t-1)}) \leq C\cdot B^{(t-1)}(e) \leq \sum_{v \in e} \phi_v. \qedhere 
	\]
\end{proof}

It remains to relate the value of the solution $M$ to the cost of this dual.
For this, we introduce the following useful notation. 
For any edge $e\in S$, we define the weight of $e$ to be
\[
w_e := \begin{cases}
	f(e:M) & e\in M \\
	f(e:M^{(t)}) & e\in M^{(t-1)}\setminus M^{(t)} \subseteq P.
\end{cases}
\]
In words, the weight of an edge in the matching is $f(e : M)$, and the weight of a preempted edge is frozen to its last value before the edge was preempted.
One simple consequence of the definition of the weights $w_e$ is the following relationship to $f(M)$.
\begin{obs}\label{obs:weight-and-f}
	$f(M) = \sum_{e \in M} f(e: M) = \sum_{e \in M} w_e = w(M)$.
\end{obs}

We now show that the preempted edges' weight is bounded in terms of the weight of $M$.
\begin{lem}\label{weight-preempted}
	The weights of $P$ and $M$ satisfy $w(P)\leq w(M)\cdot \frac{1}{C-1}$.
\end{lem}
\begin{proof}
	For any edge $e$, we define the following set of preempted edges which are preempted in favor of$e$ or in favor of an edge (recursively) preempted due to $e$.\footnote{In \cite{feigenbaum2005graph,mcgregor2005finding}, these sets are referred to by the somewhat morbid term ``trail of the dead''.}
	First, for an edge $e=e^{(t)}$, we let the set $P^1(e):=N(e)\cap M^{(t-1)}$ denote the edges preempted when $e$ is added to $M$. For any $i>1$, we let $P^i(e) := P^1(P^{i-1}(e))$ be the set of edges preempted by an edge with a trail of preemptions of length $i-1$ from $e$.
	By \Cref{line:nm-cover-check}, we have that any edge $e\in S$ has weight at least
	$w_e \geq C\cdot P^1(e)$. 
	By induction, this implies that
	$w_e \geq C\cdot w(P^{(i-1)}(e)) \geq C^i \cdot w(P^i(e))$.
	Now, since each preempted edge $e'\in P$ belongs to precisely one set $P^i(e)$ for some $i\geq 1$ and $e\in M$, we find that indeed,
	\[w(P) = \sum_{e \in M} \sum_{i\geq 1} w(P^i(e)) \leq \sum_{e \in M} \sum_{i\geq 1} \frac{1}{C^i} \cdot w_e = w(M)\cdot \left(\frac{1}{C}+\frac{1}{C^2} + \frac{1}{C^3} +\dots \right) = w(M) \cdot \frac{1}{C-1}.\qedhere\]
\end{proof}

Using \Cref{weight-preempted}, we can now relate the value of the primal solution $M$ to the cost of the our dual solution, $\mu+\sum_v \phi_v$.
We start by bounding $\mu$ in terms of $f(M)$.

\begin{lem}
	\label{cor:m_fandmu-CK}
	The matching $M$ output by \Cref{alg:MSbM-CK-FKK} satisfies $f(M) \geq  \left(1 - \frac{1}{C}\right) \cdot \mu$.
\end{lem}
\begin{proof}
	By submodularity of $f$, \Cref{weight-preempted}, and \Cref{obs:weight-and-f} we obtain the desired inequality,
	\[
	\mu = f(S) = f(M\cup P) 
	\leq f(M) + \sum_{e \in P} f(e: M) = w(M) + w(P) \leq \left(1+\frac{1}{C-1}\right)\cdot f(M). \qedhere
	\]
\end{proof}

We next bound $\sum_v \phi_v$ in terms of $f(M)$.

\begin{lem}
	\label{cor:m_fandphi-CK}
	The matching $M$ output by \Cref{alg:MSbM-CK-FKK} run with $C> 1$ satisfies $$f(M) \geq \frac{1}{2C+C/(C-1)} \cdot \sum_{v \in V} \phi_v.$$
\end{lem}
\begin{proof}	
	Fix a vertex $v$ and edge $e\in M^{(t-1)}\setminus M^{(t)} \subseteq P$ preempted at time $t$ in favor of edge $e'=e^{(t)}\ni v$. For this edge $e$, by monotonicty in $t'$ of $f(e': M^{(t')})$, the test of \Cref{line:nm-cover-check-CK}, non-negativity of $f(e'': M^{(t-1)})$ for any edge $e''\in M^{(t-1)}$ and $C>1$ we have that
	\[w_{e'} \geq f(e': M^{(t-1)}) \geq C\cdot B^{(t-1)}(e')
	\geq C\cdot f(e:M^{(t-1)}) = C\cdot w_e > w_e.\]
	Consequently, again relying on monotonicity in $t'$ of $f(e': M^{(t')})$, we have that for any edge $e\in P$, there is at most one vertex $v\in e$ such that $w_e=f(e: M^{(t-1)})$ is equal to $\phi_v = \max\{f(e': M^{(t')}) \mid v\in e'\in M^{(t')}\}$ . Edges $e\in M$, on the other hand, clearly have $w_e = f(e:M)$ equal to $\phi_v = \max\{f(e': M^{(t-1)}) \mid v\in e'\in M^{(t-1)}\setminus M^{(t)}\}$ for at most two vertices $v\in e$. Combined with \Cref{weight-preempted} and \Cref{obs:weight-and-f}, this yields the desired inequality,
	\[
	\sum_{v\in V} \phi_v \leq C\cdot \left(2\sum_{e\in M} w_e + \sum_{e \in P} w_e\right) \leq  \left(2C+\frac{C}{C-1}\right)\cdot w(M) = \left(2C+\frac{C}{C-1}\right)\cdot f(M).\qedhere
	\]
\end{proof}

Equipped with the above lemmas, we can now analyze \Cref{alg:MSbM-CK-FKK}'s approximation ratio.
\begin{thm}
	\Cref{alg:MSbM-CK-FKK} run with $C>1$ and $q=1$ on a monotone \msm instance outputs a matching $M$ of value \[\left(2C+\frac{2C}{C-1}\right) \cdot f(M) \geq f(\opt).\]
	This is optimized by taking $C=2$, resulting in an approximation ratio of $8$.
\end{thm}

\begin{proof}
	By weak LP duality and \Cref{obs:m-g-dual-covered-CK}, together with monotonicity of $f$, we have that
	\[C \cdot\sum_v  \phi_v + \mu \geq \max_T g^S(T) = \max_T f(S \cup T) \geq f(S \cup \opt) \geq f(\opt).\]
	Combining \cref{cor:m_fandphi-CK} and $f(M)=\mu$ by definition and rearranging, we get the desired inequality,
	\[\left(2C + \frac{2C}{C-1}\right) \cdot  f(M) \geq C \cdot\sum_v  \phi_v + \mu \geq f(\opt). \qedhere\]
\end{proof}

As with our algorithm of \Cref{sec:monotone}, our analysis of \Cref{alg:MSbM-CK-FKK} relied on monotonicity, crucially using $f(S\cup OPT)\geq f(OPT)$. 
To extend this algorithm to non-monotone \msm, we again appeal to \Cref{technion-lemma}, setting $q=\frac{1}{2C+1}$. This is precisely the algorithm of \citet{feldman2018less}, which we analyze in the following section.

%% file: FKK.tex
\subsection{The Algorithm of \citet{feldman2018less}}\label{sec:FKK}
In \cite{feldman2018less}, Feldman et al.~showed how to generalize the algorithm of \cite{chakrabarti2015submodular} to non-monotone function maximization. Here we show an analysis of their algorithm in our primal dual framework. Our proof is an extension of the one in \cref{sec:CK+mcG} in a way that is analogous to how \cref{sec:non-monotone} extends \cref{sec:monotone}.	

We reuse the same dual from \cref{sec:CK+mcG}, only this time, both our dual object and the function $g^S$ are random variables. The proof of expected dual feasibility for this variant of the algorithm of \Cref{sec:CK+mcG} is analogous to that of \Cref{obs:nm-g-dual-covered}, so we only outline the differences here.

We start with expected feasibility.

\begin{lem}\label{obs:FKK-dual-covered}
	The dual solution $(\expect{\vec \lambda}, \expect{\vec \phi}, \expect{\mu})$ is feasible for the expected LP $\expect{(D)}$.
\end{lem}
\begin{proof}[Proof (Sketch)]
	
	The first set of constraints is satisfied for any random realization. Indeed, as in the proof of \Cref{obs:m-g-dual-covered-CK}, for any realization of $S$, by submodularity of $f$, we have
	\begin{align*}
	\sum_{e \in T} \lambda_e = \sum_{e \in T \setminus S} f(e:S) \geq \sum_{e \in T \setminus S} f_{S}(e) \geq f_{S}(T \setminus S) = f(S \cup T) - f(S) = g^S(T) - \mu.
	\end{align*}
	Consequently, taking expectation over $S$, we have that indeed, $\E_S[\mu] + \sum_{e\in T} \E_S[\lambda_e] \geq \E_S[g^S(T)]$. 
	
	For the second set of constraints, the proof is nearly identical to that of \Cref{obs:nm-g-dual-covered}, where we show that
	\[\expect*{\sum_{v \in e} \phi_v } \geq \expect*{\lambda_e }. \]
	This is proved by taking total probability over the event $A_e := [f(e: S) \leq C\cdot \sum_{v \in V} \phi^{(t-1)}_v]$ and its complement. The key inequality to prove here is that for any realization of randomness $R$ for which $\overline{A_e}$ holds, 
	we have that
	\[
	\expect*{\sum_{v \in e} \phi_v^{(t)} | R} = 2q \cdot f(e: S) + (1-2q)\cdot \sum_{v \in e} \phi_v^{(t-1)}\geq 2q \cdot f(e:S). \]
	And indeed, conditioned on $R$, the edge $e= e^{(t)}$ fails the test in \cref{line:nm-cover-check-CK}, and so with probability $q$, we have $\sum_{v\in e}\phi_v^{(t)} = 2\cdot f(e:S)$. To see this, note that if $e$ is added to the matching, then for both $v \in e$, by definition $\phi^{(t)}_v$ must be at least $f(e : S)$. Hence, in this case
	\[
	\expect*{\sum_{v \in e} \phi_v^{(t)} | R} \geq 2q \cdot f(e: S). \]
	The proof then proceeds as that of \Cref{obs:nm-g-dual-covered}.
\end{proof}

To relate the value of the solution $M$ to the cost of the dual, we can define weights as in \cref{sec:CK+mcG} and reuse  lemmas \ref{weight-preempted}, \ref{cor:m_fandmu-CK}, and \ref{cor:m_fandphi-CK}, which hold for every realization of the random choices of the algorithm. From here, following our template, we can use these along with LP duality, \Cref{obs:FKK-dual-covered} and \Cref{technion-lemma}, to analyze this algorithm.

\begin{thm}
	\Cref{alg:MSbM-CK-FKK} run with $q=1/(2C+1)$ and $C$ on a non-monotone \msm instance outputs a matching $M$ of value \[\left(\frac{2	C^2+C}{C-1}\right) \cdot f(M) \geq f(\opt).\]
	This is optimized by taking $C=1+\frac{\sqrt{3}}{2}$, resulting in an approximation ratio of $5+2\sqrt{6}\approx 9.899$.
	Moreover, the same algorithm is $2C+2C/(C-1)$ approximate for \emph{monotone} \msm, yielding an approximation ratio of $8$ for $C=2$.
\end{thm}

%% file: tight.tex
\section{Tight instance for \cref{alg:MSbM}}\label{sec:tight-MSM}

In this section we show that there exists a family of instances of MSM instances parametrized by $C$ for which \cref{alg:MSbM} with parameter $C>1$ yields an approximation factor of $2C + C/(C-1)$.

\begin{lem}
	The approximation ratio of \Cref{alg:MSbM} with $C>1$ and $q=1$ for monotone MSM is at least $2C+\frac{C}{C-1}$.
\end{lem}
\begin{proof}
Define the graph $G$ as follows. The vertex set $V(G)$ consists of $\{x_i, y_i\}_{i \in [0,n]}$. For convenience, for every $i \in [1,n]$ we define the edges $d_i = (x_0, x_i)$  and $e_i = (x_i, y_i)$. Then the edge set $E(G)$ consists of the edges $\{d_i\}_{i=1}^n \cup \{e_i\}_{i=0}^n$.

\begin{figure}[!h]
	\centering
	\scalebox{0.85}{\input{tightInstance.tikz}}
	\caption{Tight Example for \Cref{alg:MSbM}}
\end{figure}

To define the (monotone) submodular function, we first define an auxiliary weight function $w: E(G) \rightarrow \R_{\geq 0}$. The weights are:
\begin{align*}
	w(d_i) &= C^{i-1} \tag{$n \geq i \geq 1$} \\
	w(e_1) &= 1+ C -\eps \\
	w(e_i) &= C^{i} - \eps \tag{$n \geq i \geq 2$ } \\
	w(e_0) &=  C^n -\eps \\
	\intertext{Now the submodular function is:}
	f(T) &:= w(T \cap \{e_0\}) + \sum_{i=0}^n \min(w(T \cap \{d_i, e_i\}), w(e_i))
\end{align*}
Since weights are non-negative, this function is monotone. Submodularity follows from preserevation of subdmodularity under linear combinations (and in particular sums), and $\min\{w(S),X\}$ being submodular for any linear function $w$.

The stream reveals the edges $d_1, \ldots, d_n$ in order, and subsequently reveals $e_0, e_1, \ldots, e_n$ in order. For a run of \Cref{alg:MSbM} with this choice of $C$ and $q=1$, several claims hold inductively: 
\begin{enumerate}[(a)]
	\item On the arrival of edge $d_i$, we have $\phi_{x_0} = C^{i-2}$ (except for the arrival of $d_1$, at which point $\phi_0 = 0$) and $\phi_{x_i} = 0$.
	\item The algorithm takes every edge $d_i$ into the stack.
	\item After $d_i$ is taken into the stack, we have $\phi_{x_0} = C^{i-1}$ and $\phi_{x_i} = C^{i-1} + C^{i-2}$ (except for $\phi_{x_1}$ which is set to $1$).
	\item The algorithm does not take $e_i$ into the stack.
\end{enumerate}

Let $\Lambda_t$ be the statement that these claims holds for time $t$. By inspection $\Lambda_1$ holds, now consider some time $ i > 1$. Claim (a) follows directly from claim (c) of $\Lambda_{i-1}$. Claim (b) follows from (a) since $f_S(d_i) = C^{i-1} = C\cdot \phi_0$ when $d_i$ arrives. Claim (c) is a consequence of how the algorithm increases the potentials $\phi$ when taking edges into the stack. Claim (d) holds since $f_S(e_i) = w(e_i) - w(d_i) = C^{i} - C^{i-1} - \eps < C \cdot \phi_{x_i}$.

From the above, we find that \Cref{alg:MSbM} with parameter $C$ as above and $q=1$ will have all edges $d_1,\dots,d_n$ in its stack by the end, resulting in it outputting the matching consisting of the single edge $d_n$. The value of this edge (and hence this matching) is $C^{n-1}$, while on the other hand $\opt$ can take the edges $\{e_i\}_{i=0}^n$, which have value
\begin{align*}\sum_{i=0}^n w(e_i) = C^n + \sum_{i=0}^n C^{i} - \eps(n+1) \rightarrow C^{n-1} \left( 2C +   \frac{C}{C-1}\right). \tag{as $n \rightarrow \infty$ and $\eps \rightarrow 0$}\end{align*}
so long as $C > 1$. Hence $c(\opt) / c(\alg) \rightarrow 2C + C / (C-1)$. The lemma follows.
\end{proof}

%% file: tightInstance.tikz
\tikzset{every picture/.style={line width=0.75pt}} 

\begin{tikzpicture}[x=0.75pt,y=0.75pt,yscale=-1,xscale=1]

\draw  [fill={rgb, 255:red, 0; green, 0; blue, 0 }  ,fill opacity=1 ] (169.91,124.91) .. controls (169.86,122.11) and (172.1,119.79) .. (174.91,119.74) .. controls (177.72,119.69) and (180.03,121.93) .. (180.08,124.74) .. controls (180.13,127.55) and (177.89,129.86) .. (175.09,129.91) .. controls (172.28,129.96) and (169.96,127.72) .. (169.91,124.91) -- cycle ;
\draw  [fill={rgb, 255:red, 0; green, 0; blue, 0 }  ,fill opacity=1 ] (259.74,65.09) .. controls (259.69,62.28) and (261.93,59.97) .. (264.74,59.92) .. controls (267.55,59.87) and (269.86,62.11) .. (269.91,64.91) .. controls (269.96,67.72) and (267.72,70.04) .. (264.91,70.09) .. controls (262.11,70.14) and (259.79,67.9) .. (259.74,65.09) -- cycle ;
\draw    (175,124.83) -- (264.83,65) ;
\draw  [fill={rgb, 255:red, 0; green, 0; blue, 0 }  ,fill opacity=1 ] (259.92,105.26) .. controls (259.87,102.45) and (262.11,100.14) .. (264.91,100.09) .. controls (267.72,100.04) and (270.04,102.28) .. (270.09,105.09) .. controls (270.14,107.89) and (267.9,110.21) .. (265.09,110.26) .. controls (262.28,110.31) and (259.97,108.07) .. (259.92,105.26) -- cycle ;
\draw    (175,124.83) -- (265,105.17) ;
\draw    (175,124.83) -- (65,125.17) ;
\draw  [fill={rgb, 255:red, 0; green, 0; blue, 0 }  ,fill opacity=1 ] (59.92,125.26) .. controls (59.87,122.45) and (62.11,120.14) .. (64.91,120.09) .. controls (67.72,120.04) and (70.04,122.28) .. (70.09,125.09) .. controls (70.14,127.89) and (67.9,130.21) .. (65.09,130.26) .. controls (62.28,130.31) and (59.97,128.07) .. (59.92,125.26) -- cycle ;
\draw  [fill={rgb, 255:red, 0; green, 0; blue, 0 }  ,fill opacity=1 ] (259.74,145.35) .. controls (259.69,142.54) and (261.93,140.23) .. (264.74,140.18) .. controls (267.54,140.13) and (269.86,142.37) .. (269.91,145.17) .. controls (269.96,147.98) and (267.72,150.3) .. (264.91,150.35) .. controls (262.1,150.4) and (259.79,148.16) .. (259.74,145.35) -- cycle ;
\draw    (175,124.83) -- (264.82,145.26) ;
\draw  [fill={rgb, 255:red, 0; green, 0; blue, 0 }  ,fill opacity=1 ] (259.91,244.91) .. controls (259.86,242.11) and (262.1,239.79) .. (264.91,239.74) .. controls (267.72,239.69) and (270.03,241.93) .. (270.08,244.74) .. controls (270.13,247.55) and (267.89,249.86) .. (265.09,249.91) .. controls (262.28,249.96) and (259.96,247.72) .. (259.91,244.91) -- cycle ;
\draw    (175,124.83) -- (265,244.83) ;
\draw    (374.82,64.65) -- (264.83,65) ;
\draw  [fill={rgb, 255:red, 0; green, 0; blue, 0 }  ,fill opacity=1 ] (369.73,64.74) .. controls (369.68,61.93) and (371.92,59.62) .. (374.73,59.57) .. controls (377.54,59.52) and (379.86,61.76) .. (379.91,64.57) .. controls (379.95,67.37) and (377.72,69.69) .. (374.91,69.74) .. controls (372.1,69.79) and (369.78,67.55) .. (369.73,64.74) -- cycle ;
\draw    (374.91,104.91) -- (264.92,105.26) ;
\draw  [fill={rgb, 255:red, 0; green, 0; blue, 0 }  ,fill opacity=1 ] (369.83,105) .. controls (369.78,102.19) and (372.02,99.88) .. (374.82,99.83) .. controls (377.63,99.78) and (379.95,102.02) .. (380,104.82) .. controls (380.05,107.63) and (377.81,109.95) .. (375,110) .. controls (372.19,110.05) and (369.88,107.81) .. (369.83,105) -- cycle ;
\draw    (374.91,144.91) -- (264.92,145.26) ;
\draw  [fill={rgb, 255:red, 0; green, 0; blue, 0 }  ,fill opacity=1 ] (369.83,145) .. controls (369.78,142.19) and (372.02,139.88) .. (374.82,139.83) .. controls (377.63,139.78) and (379.95,142.02) .. (380,144.82) .. controls (380.05,147.63) and (377.81,149.95) .. (375,150) .. controls (372.19,150.05) and (369.88,147.81) .. (369.83,145) -- cycle ;
\draw    (374.91,244.91) -- (264.92,245.26) ;
\draw  [fill={rgb, 255:red, 0; green, 0; blue, 0 }  ,fill opacity=1 ] (369.83,245) .. controls (369.78,242.19) and (372.02,239.88) .. (374.82,239.83) .. controls (377.63,239.78) and (379.95,242.02) .. (380,244.82) .. controls (380.05,247.63) and (377.81,249.95) .. (375,250) .. controls (372.19,250.05) and (369.88,247.81) .. (369.83,245) -- cycle ;

\draw (167,96.4) node [anchor=north west][inner sep=0.75pt]    {$x_{0}$};
\draw (315,172.4) node [anchor=north west][inner sep=0.75pt]  [font=\LARGE]  {${\displaystyle \vdots }$};
\draw (59,96.4) node [anchor=north west][inner sep=0.75pt]    {$y_{0}$};
\draw (271,41.4) node [anchor=north west][inner sep=0.75pt]    {$x_{1}$};
\draw (271,80.4) node [anchor=north west][inner sep=0.75pt]    {$x_{2}$};
\draw (271,120.4) node [anchor=north west][inner sep=0.75pt]    {$x_{2}$};
\draw (270,222.4) node [anchor=north west][inner sep=0.75pt]    {$x_{n}$};
\draw (382,41.4) node [anchor=north west][inner sep=0.75pt]    {$y_{1}$};
\draw (382,81.4) node [anchor=north west][inner sep=0.75pt]    {$y_{2}$};
\draw (382,121.4) node [anchor=north west][inner sep=0.75pt]    {$y_{2}$};
\draw (381,221.4) node [anchor=north west][inner sep=0.75pt]    {$y_{n}$};

\end{tikzpicture}

%% file: extraproofs.tex
\section{Space Bound of \Cref{alg:MSbM}}
\label{sec:deferred-msm-algo}

In this section we bound the space and time complexities of \Cref{alg:MSbM}. We start by bounding its space usage,  as restated in the following lemma.
\spacebound*

First, we prove the following simpler lemma.
\begin{lem}\label{per-node-bound}
	Each vertex $v\in V$ has at most $\tilde{O}(b_v)$ edges in the stack $S$.
\end{lem}
\begin{proof}
	If an edge $e\ni v$ is added to $S$ at time $t$, then by the test in \Cref{line:nm-cover-check}, $f(e:S)\geq (1+\eps)\cdot \sum_{u\in e} \phi^{(t-1)}_u$. Consequently, and since $\phi$ values are easily seen to always be positive, we have  \begin{align*}
		\phi^{(t)}_v-\phi^{(t-1)}_v 
		&= \frac{f(e:S)-\sum_{u\in e}\phi^{(t-1)}_u}{b_v} \\
		&\geq \frac{\eps \cdot \sum_{u\in e}\phi^{(t-1)}_u}{b_v} \geq  \frac{\eps \cdot \phi^{(t-1)}_v}{b_v}.
	\end{align*}
	Thus, adding this edge $e\ni v$ to $S$ results in $\phi^{(t)}_v \geq \phi^{(t-1)}_v\cdot (1+\epsilon/b_v)$. Moreover, if $e$ is the first edge of $v$ added to $S$, then we have
	\begin{align*}
		\phi^{(t)}_v &= \frac{f(e:S)-\sum_{u\in e}\phi^{(t-1)}_u}{b_v}\geq \frac{\eps}{1+\eps} \cdot \frac{ f(e:S)}{b_v}  \\
		&\geq  \frac{\eps}{1+\eps} \cdot \frac{ f_{\min}}{b_v}.
	\end{align*}
	Hence, if $v$ had $k$ edges added to the stack by time $t$, 
	\begin{equation}\label{phi-space-lb}
	\phi^{(t)}_v\geq (\eps/(1+\eps))\cdot (f_{min}/b_v) \cdot (1+\eps/b_v)^{k-1}.
	\end{equation}
	On the other hand, since $f$ is polynomially bounded, we have that for some constant $d$
	\begin{equation}\label{phi-space-ub}
	\phi^{(t)}_v\leq \sum_{e\ni v} f_{S_e}(e)/b_v \leq n^{d} \cdot (f_{min}/b_v).
	\end{equation}
	Combining equations \eqref{phi-space-lb} and \eqref{phi-space-ub} and simplifying, we find that $(1+\eps/b_v)^{k-1} \leq n^d \cdot (1+\eps)/\epsilon.$
	Taking out logarithms and simplifying further, we find that 
	\begin{align*}
	k \leq 1+\frac{d\log n + \log(1+\eps) + \log (1/\eps)}{\log(1+\epsilon/b_v)} = O((b_v/\epsilon) \cdot (\log n + \log (1/\epsilon))) = \tilde{O}(b_v).
	\end{align*}
	That is, the number of edges of $v$ in the stack is at most $\tilde{O}(b_v)$.
\end{proof}

\begin{proof}[Proof of \Cref{space-bound}]
	Consider a maximum cardinality $b$-matching $M$ (i.e., $|M|=M_{\max}$).
	Denote by $d_v$ the degree of a vertex $v$ in $G$. Then, the set $U$ of saturated vertices $v\in V$ in $M$, i.e., those with $\min\{b_v,d_v\}$ edges in $M$, is a vertex cover.
	To see this, we note that no edge $e=(u,v)$ can have both endpoints not saturated by $M$, as the converse would imply that $M\cup \{e\}$ is a feasible $b$-matching, contradicting maximality of $M$.
	Now, since each edge of $M$ has at most two saturated endpoints, we have that $\sum_{u\in U} \min\{b_u,d_u\} \leq 2|M|$. 
	On the other hand, since $U$ is a vertex cover, the number of edges in $S$ is upper bounded by the number of edges of vertices in $U$ in $S$. 
	But each vertex $u\in U$ trivially has at most $d_u$ edges in $S$. Combined with 
	\Cref{per-node-bound}, we find that the number of edges in $S$ is at most	
	\begin{align*}
	|S| & = \tilde{O}\left(\sum_{u \in U} \min\{b_u,d_u\}\right) = \tilde{O}(|M|).
	\qedhere
	\end{align*}
	Finally, since each edge added to $S$ causes at most two nodes to have non-zero $\phi$ value, and since each such value can be specified using $\tilde{O}(1)$ bits, by the assumption of $f$ being poly-bounded, the space to store all non-zero $\phi$ values is at most $\tilde{O}(|S|) = \tilde{O}(|M|) = \tilde{O}(M_{\max})$.
\end{proof}

Finally, we briefly analyze the algorithm's running time.

\algtime*

\begin{proof}
	The algorithm clearly requires a constant number of evaluations of $\phi_v$ and $f$ for each edge arrival. The difference between deterministic and randomized implementations is due to the data structures used to maintain the mapping from $v$ to $\phi_v$ (if non zero). 
	Finally, the preprocessing time follows directly from our upper bound on $|S|$ and the post-processing stage clearly taking time linear in $|S|$.
\end{proof}

\section{Deferred Proofs of \Cref{sec:MSM-lb}}
\label{sec:extraproofs}

\maxkcov*
	
	\begin{proof}
		The first statement is precisely Corollary 1.6 of \cite{dinur2014analytical}.
		
		For the extra assumptions, if $\scopt < |\scS|^{\gamma \alpha}$ then the brute force algorithm that checks all subsets of size $\scopt$ runs in time $|\scS|^{\scopt} < 2^{|\scS|^{\gamma \alpha} \log |\scS|} \leq  2^{N^{\gamma \alpha}}$. If $|\scS| < |\scU|^{\gamma \alpha}$ , then one can brute force over all sub collections of $S$ in time $2^{|\scS|} \leq 2^{|\scU|^{\gamma \alpha}} \leq 2^{N^{\gamma \alpha}}$. Both running times contradict \cref{lem:sc_hardness}.
	\end{proof}

\gkk*

\begin{proof}
	Let $\mathcal{A}$ be an algorithm that outputs $\gamma n$ of the edges of $M_r'$ at the end of phase 1 that uses fewer than $s= n \poly \log n$ bits. We will show that $\gamma = o(1)$. 
	
	Let $\mathcal{G}$ be the set of possible first phase graphs. Then
	\[|\mathcal{G}| = \binom{n/2}{\delta n}^t = 2^{\gamma m} \]
	for some $\gamma > 0$. Let $\phi: \mathcal{G} \rightarrow \{0,1\}^s$ be the function that takes an input graph $G$ to the state of the algorithm $\mathcal{A}$ after running $\mathcal{A}$ on $G$. Let $\Gamma(G) = \{H \mid \phi(G)  = \phi(H)\}$, that is the set of graphs  inducing the same internal state for $\mathcal{A}$ at the end of phase 1. 
	
	Define $\Psi(G) =\bigcap_{H \in \Gamma(G)} E(H)$. Note that for any input graph $G$, the algorithm $\mathcal{A}$ can output an edge $e$ if and only if $e \in \Psi(G)$. Also, for any $G$ let $t'$ be the number of matchings in the RS graph $G_0$ for which $\Psi(G)$ contains at least $\gamma n$ edges. Since algorithm $\mathcal{A}$ outputs $\gamma n$ edges of $M_r'$,  the number of graphs in $\Gamma(G)$ is bounded by
	\[\binom{(1/2 - \gamma)n}{\delta n}^{t'} \binom{n/2}{\delta n}^{t-t'} = \left(2^{-\Omega(\gamma n)} \binom{n/2}{\delta n} \right)^{t'} \binom{n/2}{\delta n}^{t-t'} = 2^{-\Omega(t' \gamma n)} 2^{\gamma m} \tag{$*$} \label{eq:graphfamilybound} \]
	
	On the other hand, since the first phase graph $G$ is chosen uniformly at random, by a counting argument, with probability at least $1-o(1)$ we have that $|\Gamma(G)| \geq 2^{(\gamma - o(1))m}$. Conditioning on this happening, we also know that $t' \geq \Omega(t)$ since the input graph is uniformly chosen within $\Gamma(G)$, and the algorithm succeeds with constant probability. These two facts together with \eqref{eq:graphfamilybound} imply that $\gamma = o(1)$.
\end{proof}

%% file: main.bbl
\begin{thebibliography}{}

\bibitem[\protect\citeauthoryear{}{MSM}{}]{MSMopenQ}
Bertinoro workshop 2014, problem 63.
\newblock \url{https://sublinear.info/index.php?title=Open_Problems:63}.
\newblock Accessed: 2020-06-20.

\bibitem[\protect\citeauthoryear{Agrawal, Gollapudi, Halverson, and
  Ieong}{Agrawal et~al\mbox{.}}{2009}]{agrawal2009diversifying}
{\sc Agrawal, R.}, {\sc Gollapudi, S.}, {\sc Halverson, A.}, {\sc and} {\sc
  Ieong, S.} 2009.
\newblock Diversifying search results.
\newblock In {\em Proceedings of 2nd International Conference on Web Search and
  Data Mining (WSDM)}. 5--14.

\bibitem[\protect\citeauthoryear{Ahmed, Dickerson, and Fuge}{Ahmed
  et~al\mbox{.}}{2017}]{ahmed2017diverse}
{\sc Ahmed, F.}, {\sc Dickerson, J.~P.}, {\sc and} {\sc Fuge, M.} 2017.
\newblock Diverse weighted bipartite b-matching.
\newblock In {\em Proceedings of 26th International Joint Conference on
  Artificial Intelligence (IJCAI)}. 35--41.

\bibitem[\protect\citeauthoryear{Alaluf, Ene, Feldman, Nguyen, and Suh}{Alaluf
  et~al\mbox{.}}{2020}]{alaluf2020optimal}
{\sc Alaluf, N.}, {\sc Ene, A.}, {\sc Feldman, M.}, {\sc Nguyen, H.~L.}, {\sc
  and} {\sc Suh, A.} 2020.
\newblock Optimal streaming algorithms for submodular maximization with
  cardinality constraints.
\newblock In {\em Proceedings of 47th International Colloquium on Automata,
  Languages and Programming (ICALP)}. 6:1--6:19.

\bibitem[\protect\citeauthoryear{Anstee}{Anstee}{1987}]{anstee1987polynomial}
{\sc Anstee, R.~P.} 1987.
\newblock A polynomial algorithm for b-matchings: an alternative approach.
\newblock {\em Inf. Process. Lett.\/}~{\em 24,\/}~3, 153--157.

\bibitem[\protect\citeauthoryear{Arar, Chechik, Cohen, Stein, and Wajc}{Arar
  et~al\mbox{.}}{2018}]{arar2018dynamic}
{\sc Arar, M.}, {\sc Chechik, S.}, {\sc Cohen, S.}, {\sc Stein, C.}, {\sc and}
  {\sc Wajc, D.} 2018.
\newblock Dynamic matching: Reducing integral algorithms to
  approximately-maximal fractional algorithms.
\newblock In {\em Proceedings of 45th International Colloquium on Automata,
  Languages and Programming (ICALP)}. 7:1--7:16.

\bibitem[\protect\citeauthoryear{Badanidiyuru, Mirzasoleiman, Karbasi, and
  Krause}{Badanidiyuru et~al\mbox{.}}{2014}]{badanidiyuru2014streaming}
{\sc Badanidiyuru, A.}, {\sc Mirzasoleiman, B.}, {\sc Karbasi, A.}, {\sc and}
  {\sc Krause, A.} 2014.
\newblock Streaming submodular maximization: Massive data summarization on the
  fly.
\newblock In {\em Proceedings of 20th International Conference on Knowledge
  Discovery and Data Mining (KDD)}. 671--680.

\bibitem[\protect\citeauthoryear{Bar-Noy, Bar-Yehuda, Freund, Naor, and
  Schieber}{Bar-Noy et~al\mbox{.}}{2001}]{bar2001unified}
{\sc Bar-Noy, A.}, {\sc Bar-Yehuda, R.}, {\sc Freund, A.}, {\sc Naor, J.}, {\sc
  and} {\sc Schieber, B.} 2001.
\newblock A unified approach to approximating resource allocation and
  scheduling.
\newblock {\em Journal of the ACM (JACM)\/}~{\em 48,\/}~5, 1069--1090.

\bibitem[\protect\citeauthoryear{Bar-Yehuda and Rawitz}{Bar-Yehuda and
  Rawitz}{2005}]{bar2005equivalence}
{\sc Bar-Yehuda, R.} {\sc and} {\sc Rawitz, D.} 2005.
\newblock On the equivalence between the primal-dual schema and the local ratio
  technique.
\newblock {\em SIAM Journal on Discrete Mathematics\/}~{\em 19,\/}~3, 762--797.

\bibitem[\protect\citeauthoryear{Bernstein and Stein}{Bernstein and
  Stein}{2015}]{bernstein2015fully}
{\sc Bernstein, A.} {\sc and} {\sc Stein, C.} 2015.
\newblock Fully dynamic matching in bipartite graphs.
\newblock In {\em Proceedings of 42nd International Colloquium on Automata,
  Languages and Programming (ICALP)}. 167--179.

\bibitem[\protect\citeauthoryear{Bernstein and Stein}{Bernstein and
  Stein}{2016}]{bernstein2016faster}
{\sc Bernstein, A.} {\sc and} {\sc Stein, C.} 2016.
\newblock Faster fully dynamic matchings with small approximation ratios.
\newblock In {\em Proceedings of 27th Annual ACM-SIAM Symposium on Discrete
  Algorithms (SODA)}. 692--711.

\bibitem[\protect\citeauthoryear{Bhattacharya, Henzinger, and
  Italiano}{Bhattacharya et~al\mbox{.}}{2015}]{bhattacharya2015deterministic}
{\sc Bhattacharya, S.}, {\sc Henzinger, M.}, {\sc and} {\sc Italiano, G.~F.}
  2015.
\newblock Deterministic fully dynamic data structures for vertex cover and
  matching.
\newblock In {\em Proceedings of 26th Annual ACM-SIAM Symposium on Discrete
  Algorithms (SODA)}. 785--804.

\bibitem[\protect\citeauthoryear{Buchbinder and Feldman}{Buchbinder and
  Feldman}{2019}]{buchbinder2019constrained}
{\sc Buchbinder, N.} {\sc and} {\sc Feldman, M.} 2019.
\newblock Constrained submodular maximization via a nonsymmetric technique.
\newblock {\em Mathematics of Operations Research\/}~{\em 44,\/}~3, 988--1005.

\bibitem[\protect\citeauthoryear{Buchbinder, Feldman, Naor, and
  Schwartz}{Buchbinder et~al\mbox{.}}{2014}]{buchbinder2014submodular}
{\sc Buchbinder, N.}, {\sc Feldman, M.}, {\sc Naor, J.}, {\sc and} {\sc
  Schwartz, R.} 2014.
\newblock Submodular maximization with cardinality constraints.
\newblock In {\em Proceedings of 25th Annual ACM-SIAM Symposium on Discrete
  Algorithms (SODA)}. 1433--1452.

\bibitem[\protect\citeauthoryear{Calinescu, Chekuri, P{\'a}l, and
  Vondr{\'a}k}{Calinescu et~al\mbox{.}}{2011}]{calinescu2011maximizing}
{\sc Calinescu, G.}, {\sc Chekuri, C.}, {\sc P{\'a}l, M.}, {\sc and} {\sc
  Vondr{\'a}k, J.} 2011.
\newblock Maximizing a monotone submodular function subject to a matroid
  constraint.
\newblock {\em SIAM Journal on Computing (SICOMP)\/}~{\em 40,\/}~6, 1740--1766.

\bibitem[\protect\citeauthoryear{Chakrabarti and Kale}{Chakrabarti and
  Kale}{2015}]{chakrabarti2015submodular}
{\sc Chakrabarti, A.} {\sc and} {\sc Kale, S.} 2015.
\newblock Submodular maximization meets streaming: Matchings, matroids, and
  more.
\newblock {\em Mathematical Programming\/}~{\em 154,\/}~1-2, 225--247.

\bibitem[\protect\citeauthoryear{Chekuri, Gupta, and Quanrud}{Chekuri
  et~al\mbox{.}}{2015}]{chekuri2015streaming}
{\sc Chekuri, C.}, {\sc Gupta, S.}, {\sc and} {\sc Quanrud, K.} 2015.
\newblock Streaming algorithms for submodular function maximization.
\newblock In {\em Proceedings of 42nd International Colloquium on Automata,
  Languages and Programming (ICALP)}. 318--330.

\bibitem[\protect\citeauthoryear{Crouch and Stubbs}{Crouch and
  Stubbs}{2014}]{crouch2014improved}
{\sc Crouch, M.} {\sc and} {\sc Stubbs, D.~M.} 2014.
\newblock Improved streaming algorithms for weighted matching, via unweighted
  matching.
\newblock In {\em Proceedings of 17th International Workshop on Approximation
  Algorithms for Combinatorial Optimization Problems (APPROX)}. 96.

\bibitem[\protect\citeauthoryear{Devanur, Jain, and Kleinberg}{Devanur
  et~al\mbox{.}}{2013}]{devanur2013randomized}
{\sc Devanur, N.~R.}, {\sc Jain, K.}, {\sc and} {\sc Kleinberg, R.~D.} 2013.
\newblock Randomized primal-dual analysis of ranking for online bipartite
  matching.
\newblock In {\em Proceedings of 24th Annual ACM-SIAM Symposium on Discrete
  Algorithms (SODA)}. 101--107.

\bibitem[\protect\citeauthoryear{Dickerson, Sankararaman, Srinivasan, and
  Xu}{Dickerson et~al\mbox{.}}{2019}]{dickerson2019balancing}
{\sc Dickerson, J.~P.}, {\sc Sankararaman, K.~A.}, {\sc Srinivasan, A.}, {\sc
  and} {\sc Xu, P.} 2019.
\newblock Balancing relevance and diversity in online bipartite matching via
  submodularity.
\newblock In {\em Proceedings of 53rd AAAI Conference on Artificial
  Intelligence (AAAI)}. Vol.~33. 1877--1884.

\bibitem[\protect\citeauthoryear{Dinur and Steurer}{Dinur and
  Steurer}{2014}]{dinur2014analytical}
{\sc Dinur, I.} {\sc and} {\sc Steurer, D.} 2014.
\newblock Analytical approach to parallel repetition.
\newblock In {\em Proceedings of 46th Annual ACM Symposium on Theory of
  Computing (STOC)}. 624--633.

\bibitem[\protect\citeauthoryear{Ene and Nguyen}{Ene and
  Nguyen}{2016}]{ene2016constrained}
{\sc Ene, A.} {\sc and} {\sc Nguyen, H.~L.} 2016.
\newblock Constrained submodular maximization: Beyond 1/e.
\newblock In {\em Proceedings of 57th Symposium on Foundations of Computer
  Science (FOCS)}. 248--257.

\bibitem[\protect\citeauthoryear{Epstein, Levin, Segev, and Weimann}{Epstein
  et~al\mbox{.}}{2013}]{epstein2013improved}
{\sc Epstein, L.}, {\sc Levin, A.}, {\sc Segev, D.}, {\sc and} {\sc Weimann,
  O.} 2013.
\newblock Improved bounds for online preemptive matching.
\newblock In {\em Proceedings of 30th International Symposium on Theoretical
  Aspects of Computer Science (STACS)}. 389.

\bibitem[\protect\citeauthoryear{Fahrbach, Huang, Tao, and
  Zadimoghaddam}{Fahrbach et~al\mbox{.}}{2020}]{fahrbach2020edge}
{\sc Fahrbach, M.}, {\sc Huang, Z.}, {\sc Tao, R.}, {\sc and} {\sc
  Zadimoghaddam, M.} 2020.
\newblock Edge-weighted online bipartite matching.
\newblock In {\em Proceedings of 61st Symposium on Foundations of Computer
  Science (FOCS)}.
\newblock To appear.

\bibitem[\protect\citeauthoryear{Feige}{Feige}{1998}]{feige1998threshold}
{\sc Feige, U.} 1998.
\newblock A threshold of ln n for approximating set cover.
\newblock {\em Journal of the ACM (JACM)\/}~{\em 45,\/}~4, 634--652.

\bibitem[\protect\citeauthoryear{Feigenbaum, Kannan, McGregor, Suri, and
  Zhang}{Feigenbaum et~al\mbox{.}}{2005}]{feigenbaum2005graph}
{\sc Feigenbaum, J.}, {\sc Kannan, S.}, {\sc McGregor, A.}, {\sc Suri, S.},
  {\sc and} {\sc Zhang, J.} 2005.
\newblock On graph problems in a semi-streaming model.
\newblock {\em Theor. Comput. Sci.\/}~{\em 348,\/}~2-3, 207--216.

\bibitem[\protect\citeauthoryear{Feldman, Karbasi, and Kazemi}{Feldman
  et~al\mbox{.}}{2018}]{feldman2018less}
{\sc Feldman, M.}, {\sc Karbasi, A.}, {\sc and} {\sc Kazemi, E.} 2018.
\newblock Do less, get more: streaming submodular maximization with
  subsampling.
\newblock In {\em Proceedings of 31st Annual Conference on Neural Information
  Processing Systems (NeurIPS)}. 732--742.

\bibitem[\protect\citeauthoryear{Feldman, Naor, and Schwartz}{Feldman
  et~al\mbox{.}}{2011a}]{feldman2011unified}
{\sc Feldman, M.}, {\sc Naor, J.}, {\sc and} {\sc Schwartz, R.} 2011a.
\newblock A unified continuous greedy algorithm for submodular maximization.
\newblock In {\em Proceedings of 52nd Symposium on Foundations of Computer
  Science (FOCS)}. 570--579.

\bibitem[\protect\citeauthoryear{Feldman, Naor, Schwartz, and Ward}{Feldman
  et~al\mbox{.}}{2011b}]{feldman2011improved}
{\sc Feldman, M.}, {\sc Naor, J.~S.}, {\sc Schwartz, R.}, {\sc and} {\sc Ward,
  J.} 2011b.
\newblock Improved approximations for k-exchange systems.
\newblock In {\em Proceedings of 19th Annual European Symposium on Algorithms
  (ESA)}. 784--798.

\bibitem[\protect\citeauthoryear{Feldman, Norouzi-Fard, Svensson, and
  Zenklusen}{Feldman et~al\mbox{.}}{2020}]{feldman2020one}
{\sc Feldman, M.}, {\sc Norouzi-Fard, A.}, {\sc Svensson, O.}, {\sc and} {\sc
  Zenklusen, R.} 2020.
\newblock The one-way communication complexity of submodular maximization with
  applications to streaming and robustness.
\newblock In {\em Proceedings of 52nd Annual ACM Symposium on Theory of
  Computing (STOC)}. 1363--1374.

\bibitem[\protect\citeauthoryear{Fisher, Nemhauser, and Wolsey}{Fisher
  et~al\mbox{.}}{1978}]{fisher1978analysis}
{\sc Fisher, M.~L.}, {\sc Nemhauser, G.~L.}, {\sc and} {\sc Wolsey, L.~A.}
  1978.
\newblock An analysis of approximations for maximizing submodular set
  functions—ii.
\newblock In {\em Polyhedral combinatorics}. Springer, 73--87.

\bibitem[\protect\citeauthoryear{Gabow}{Gabow}{2018}]{gabow2018data}
{\sc Gabow, H.~N.} 2018.
\newblock Data structures for weighted matching and extensions to b-matching
  and f-factors.
\newblock {\em Transactions on Algorithms (TALG)\/}~{\em 14,\/}~3, 1--80.

\bibitem[\protect\citeauthoryear{Ghaffari and Wajc}{Ghaffari and
  Wajc}{2019}]{ghaffari2019simplified}
{\sc Ghaffari, M.} {\sc and} {\sc Wajc, D.} 2019.
\newblock Simplified and space-optimal semi-streaming $(2+
  \epsilon)$-approximate matching.
\newblock In {\em Proceedings of 2nd Symposium of Simplicity in Algorithms
  (SOSA)}.

\bibitem[\protect\citeauthoryear{Goel, Kapralov, and Khanna}{Goel
  et~al\mbox{.}}{2012}]{goel2012communication}
{\sc Goel, A.}, {\sc Kapralov, M.}, {\sc and} {\sc Khanna, S.} 2012.
\newblock On the communication and streaming complexity of maximum bipartite
  matching.
\newblock In {\em Proceedings of 23rd Annual ACM-SIAM Symposium on Discrete
  Algorithms (SODA)}. 468--485.

\bibitem[\protect\citeauthoryear{Gupta, Krishnaswamy, Kumar, and
  Panigrahi}{Gupta et~al\mbox{.}}{2017}]{gupta2017online}
{\sc Gupta, A.}, {\sc Krishnaswamy, R.}, {\sc Kumar, A.}, {\sc and} {\sc
  Panigrahi, D.} 2017.
\newblock Online and dynamic algorithms for set cover.
\newblock In {\em Proceedings of 49th Annual ACM Symposium on Theory of
  Computing (STOC)}. 537--550.

\bibitem[\protect\citeauthoryear{Gupta and Levin}{Gupta and
  Levin}{2020a}]{gupta2020fully}
{\sc Gupta, A.} {\sc and} {\sc Levin, R.} 2020a.
\newblock Fully-dynamic submodular cover with bounded recourse.
\newblock In {\em Proceedings of 61st Symposium on Foundations of Computer
  Science (FOCS)}. To appear.

\bibitem[\protect\citeauthoryear{Gupta and Levin}{Gupta and
  Levin}{2020b}]{gupta2020online}
{\sc Gupta, A.} {\sc and} {\sc Levin, R.} 2020b.
\newblock The online submodular cover problem.
\newblock In {\em Proceedings of 31st Annual ACM-SIAM Symposium on Discrete
  Algorithms (SODA)}. 1525--1537.

\bibitem[\protect\citeauthoryear{Gupta, Roth, Schoenebeck, and Talwar}{Gupta
  et~al\mbox{.}}{2010}]{gupta2010constrained}
{\sc Gupta, A.}, {\sc Roth, A.}, {\sc Schoenebeck, G.}, {\sc and} {\sc Talwar,
  K.} 2010.
\newblock Constrained non-monotone submodular maximization: Offline and
  secretary algorithms.
\newblock In {\em Proceedings of 6th Conference on Web and Internet Economics
  (WINE)}. 246--257.

\bibitem[\protect\citeauthoryear{Gupta and Peng}{Gupta and
  Peng}{2013}]{gupta2013fully}
{\sc Gupta, M.} {\sc and} {\sc Peng, R.} 2013.
\newblock Fully dynamic $(1+\eps)$-approximate matchings.
\newblock In {\em Proceedings of 54th Symposium on Foundations of Computer
  Science (FOCS)}. 548--557.

\bibitem[\protect\citeauthoryear{Huang, Kang, Tang, Wu, Zhang, and Zhu}{Huang
  et~al\mbox{.}}{2018a}]{huang2018match}
{\sc Huang, Z.}, {\sc Kang, N.}, {\sc Tang, Z.~G.}, {\sc Wu, X.}, {\sc Zhang,
  Y.}, {\sc and} {\sc Zhu, X.} 2018a.
\newblock How to match when all vertices arrive online.
\newblock In {\em Proceedings of 50th Annual ACM Symposium on Theory of
  Computing (STOC)}. 17--29.

\bibitem[\protect\citeauthoryear{Huang, Peng, Tang, Tao, Wu, and Zhang}{Huang
  et~al\mbox{.}}{2019}]{huang2019tight}
{\sc Huang, Z.}, {\sc Peng, B.}, {\sc Tang, Z.~G.}, {\sc Tao, R.}, {\sc Wu,
  X.}, {\sc and} {\sc Zhang, Y.} 2019.
\newblock Tight competitive ratios of classic matching algorithms in the fully
  online model.
\newblock In {\em Proceedings of 30th Annual ACM-SIAM Symposium on Discrete
  Algorithms (SODA)}. 2875--2886.

\bibitem[\protect\citeauthoryear{Huang, Tang, Wu, and Zhang}{Huang
  et~al\mbox{.}}{2018b}]{huang2018online}
{\sc Huang, Z.}, {\sc Tang, Z.~G.}, {\sc Wu, X.}, {\sc and} {\sc Zhang, Y.}
  2018b.
\newblock Online vertex-weighted bipartite matching: Beating 1-1/e with random
  arrivals.
\newblock In {\em Proceedings of 45th International Colloquium on Automata,
  Languages and Programming (ICALP)}. 1070--1081.

\bibitem[\protect\citeauthoryear{Huang, Tang, Wu, and Zhang}{Huang
  et~al\mbox{.}}{2020a}]{huang2020fully}
{\sc Huang, Z.}, {\sc Tang, Z.~G.}, {\sc Wu, X.}, {\sc and} {\sc Zhang, Y.}
  2020a.
\newblock Fully online matching ii: Beating ranking and water-filling.
\newblock In {\em Proceedings of 61st Symposium on Foundations of Computer
  Science (FOCS)}.
\newblock To appear.

\bibitem[\protect\citeauthoryear{Huang and Zhang}{Huang and
  Zhang}{2020}]{huang2020online}
{\sc Huang, Z.} {\sc and} {\sc Zhang, Q.} 2020.
\newblock Online primal dual meets online matching with stochastic rewards:
  configuration lp to the rescue.
\newblock In {\em Proceedings of 52nd Annual ACM Symposium on Theory of
  Computing (STOC)}. 1153--1164.

\bibitem[\protect\citeauthoryear{Huang, Zhang, and Zhang}{Huang
  et~al\mbox{.}}{2020b}]{huang2020adwords}
{\sc Huang, Z.}, {\sc Zhang, Q.}, {\sc and} {\sc Zhang, Y.} 2020b.
\newblock Adwords in a panorama.
\newblock In {\em Proceedings of 61st Symposium on Foundations of Computer
  Science (FOCS)}. To appear.

\bibitem[\protect\citeauthoryear{Impagliazzo and Paturi}{Impagliazzo and
  Paturi}{2001}]{impagliazzo2001complexity}
{\sc Impagliazzo, R.} {\sc and} {\sc Paturi, R.} 2001.
\newblock On the complexity of k-sat.
\newblock {\em Journal of Computer and System Sciences\/}~{\em 62,\/}~2,
  367--375.

\bibitem[\protect\citeauthoryear{Kapralov}{Kapralov}{2013}]{kapralov2013better}
{\sc Kapralov, M.} 2013.
\newblock Better bounds for matchings in the streaming model.
\newblock In {\em Proceedings of 24th Annual ACM-SIAM Symposium on Discrete
  Algorithms (SODA)}. 1679--1697.

\bibitem[\protect\citeauthoryear{Kazemi, Mitrovic, Zadimoghaddam, Lattanzi, and
  Karbasi}{Kazemi et~al\mbox{.}}{2019}]{kazemi2019submodular}
{\sc Kazemi, E.}, {\sc Mitrovic, M.}, {\sc Zadimoghaddam, M.}, {\sc Lattanzi,
  S.}, {\sc and} {\sc Karbasi, A.} 2019.
\newblock Submodular streaming in all its glory: Tight approximation, minimum
  memory and low adaptive complexity.
\newblock In {\em Proceedings of 36th International Conference on Machine
  Learning (ICML)}. 3311--3320.

\bibitem[\protect\citeauthoryear{Korula, Mirrokni, and Zadimoghaddam}{Korula
  et~al\mbox{.}}{2018}]{korula2018online}
{\sc Korula, N.}, {\sc Mirrokni, V.}, {\sc and} {\sc Zadimoghaddam, M.} 2018.
\newblock Online submodular welfare maximization: Greedy beats 1/2 in random
  order.
\newblock {\em SIAM Journal on Computing (SICOMP)\/}~{\em 47,\/}~3, 1056--1086.

\bibitem[\protect\citeauthoryear{Lattanzi, Mitrovi{\'c}, Norouzi-Fard,
  Tarnawski, and Zadimoghaddam}{Lattanzi
  et~al\mbox{.}}{2020}]{lattanzi2020fully}
{\sc Lattanzi, S.}, {\sc Mitrovi{\'c}, S.}, {\sc Norouzi-Fard, A.}, {\sc
  Tarnawski, J.}, {\sc and} {\sc Zadimoghaddam, M.} 2020.
\newblock Fully dynamic algorithm for constrained submodular optimization.
\newblock In {\em Proceedings of 34th Annual Conference on Neural Information
  Processing Systems (NeurIPS)}. To appear.

\bibitem[\protect\citeauthoryear{Lee, Mirrokni, Nagarajan, and Sviridenko}{Lee
  et~al\mbox{.}}{2009}]{lee2009non}
{\sc Lee, J.}, {\sc Mirrokni, V.~S.}, {\sc Nagarajan, V.}, {\sc and} {\sc
  Sviridenko, M.} 2009.
\newblock Non-monotone submodular maximization under matroid and knapsack
  constraints.
\newblock In {\em Proceedings of 41st Annual ACM Symposium on Theory of
  Computing (STOC)}. 323--332.

\bibitem[\protect\citeauthoryear{Lee, Mirrokni, Nagarajan, and Sviridenko}{Lee
  et~al\mbox{.}}{2010a}]{lee2010maximizing}
{\sc Lee, J.}, {\sc Mirrokni, V.~S.}, {\sc Nagarajan, V.}, {\sc and} {\sc
  Sviridenko, M.} 2010a.
\newblock Maximizing nonmonotone submodular functions under matroid or knapsack
  constraints.
\newblock {\em SIAM Journal on Discrete Mathematics\/}~{\em 23,\/}~4,
  2053--2078.

\bibitem[\protect\citeauthoryear{Lee, Sviridenko, and Vondr{\'a}k}{Lee
  et~al\mbox{.}}{2010b}]{lee2010submodular}
{\sc Lee, J.}, {\sc Sviridenko, M.}, {\sc and} {\sc Vondr{\'a}k, J.} 2010b.
\newblock Submodular maximization over multiple matroids via generalized
  exchange properties.
\newblock {\em Mathematics of Operations Research\/}~{\em 35,\/}~4, 795--806.

\bibitem[\protect\citeauthoryear{Lehmann, Lehmann, and Nisan}{Lehmann
  et~al\mbox{.}}{2006}]{lehmann2006combinatorial}
{\sc Lehmann, B.}, {\sc Lehmann, D.}, {\sc and} {\sc Nisan, N.} 2006.
\newblock Combinatorial auctions with decreasing marginal utilities.
\newblock {\em Games and Economic Behavior\/}~{\em 55,\/}~2, 270--296.

\bibitem[\protect\citeauthoryear{Lin and Bilmes}{Lin and
  Bilmes}{2011}]{lin2011word}
{\sc Lin, H.} {\sc and} {\sc Bilmes, J.} 2011.
\newblock Word alignment via submodular maximization over matroids.
\newblock In {\em Proceedings of 49th Annual Meeting of the Association for
  Computational Linguistics: Human Language Technologies (ACL)}. 170--175.

\bibitem[\protect\citeauthoryear{Manurangsi}{Manurangsi}{2020}]{manurangsi2020tight}
{\sc Manurangsi, P.} 2020.
\newblock Tight running time lower bounds for strong inapproximability of
  maximum k-coverage, unique set cover and related problems (via t-wise
  agreement testing theorem).
\newblock In {\em Proceedings of 31st Annual ACM-SIAM Symposium on Discrete
  Algorithms (SODA)}. 62--81.

\bibitem[\protect\citeauthoryear{McGregor}{McGregor}{2005}]{mcgregor2005finding}
{\sc McGregor, A.} 2005.
\newblock Finding graph matchings in data streams.
\newblock In {\em Proceedings of 8th International Workshop on Approximation
  Algorithms for Combinatorial Optimization Problems (APPROX)}. 170--181.

\bibitem[\protect\citeauthoryear{Mirzasoleiman, Jegelka, and
  Krause}{Mirzasoleiman et~al\mbox{.}}{}]{mirzasoleiman2018streaming}
{\sc Mirzasoleiman, B.}, {\sc Jegelka, S.}, {\sc and} {\sc Krause, A.}
\newblock Streaming non-monotone submodular maximization: Personalized video
  summarization on the fly.
\newblock In {\em Proceedings of 32nd AAAI Conference on Artificial
  Intelligence (AAAI)}. 1379--1386.

\bibitem[\protect\citeauthoryear{Nemhauser and Wolsey}{Nemhauser and
  Wolsey}{1978}]{nemhauser1978best}
{\sc Nemhauser, G.~L.} {\sc and} {\sc Wolsey, L.~A.} 1978.
\newblock Best algorithms for approximating the maximum of a submodular set
  function.
\newblock {\em Mathematics of operations research\/}~{\em 3,\/}~3, 177--188.

\bibitem[\protect\citeauthoryear{Nemhauser, Wolsey, and Fisher}{Nemhauser
  et~al\mbox{.}}{1978}]{nemhauser1978analysis}
{\sc Nemhauser, G.~L.}, {\sc Wolsey, L.~A.}, {\sc and} {\sc Fisher, M.~L.}
  1978.
\newblock An analysis of approximations for maximizing submodular set
  functions—i.
\newblock {\em Mathematical programming\/}~{\em 14,\/}~1, 265--294.

\bibitem[\protect\citeauthoryear{Norouzi-Fard, Tarnawski, Mitrovic, Zandieh,
  Mousavifar, and Svensson}{Norouzi-Fard
  et~al\mbox{.}}{2018}]{norouzi2018beyond}
{\sc Norouzi-Fard, A.}, {\sc Tarnawski, J.}, {\sc Mitrovic, S.}, {\sc Zandieh,
  A.}, {\sc Mousavifar, A.}, {\sc and} {\sc Svensson, O.} 2018.
\newblock Beyond 1/2-approximation for submodular maximization on massive data
  streams.
\newblock In {\em Proceedings of 35th International Conference on Machine
  Learning (ICML)}. 3829--3838.

\bibitem[\protect\citeauthoryear{Oveis~Gharan and Vondr{\'a}k}{Oveis~Gharan and
  Vondr{\'a}k}{2011}]{oveisgharan2011submodular}
{\sc Oveis~Gharan, S.} {\sc and} {\sc Vondr{\'a}k, J.} 2011.
\newblock Submodular maximization by simulated annealing.
\newblock In {\em Proceedings of 22nd Annual ACM-SIAM Symposium on Discrete
  Algorithms (SODA)}. 1098--1116.

\bibitem[\protect\citeauthoryear{Paz and Schwartzman}{Paz and
  Schwartzman}{2018}]{paz20182+}
{\sc Paz, A.} {\sc and} {\sc Schwartzman, G.} 2018.
\newblock A $(2+\epsilon)$-approximation for maximum weight matching in the
  semi-streaming model.
\newblock {\em Transactions on Algorithms (TALG)\/}~{\em 15,\/}~2, 18.

\bibitem[\protect\citeauthoryear{Peleg and Solomon}{Peleg and
  Solomon}{2016}]{peleg2016dynamic}
{\sc Peleg, D.} {\sc and} {\sc Solomon, S.} 2016.
\newblock Dynamic $(1+\epsilon)$-approximate matchings: a density-sensitive
  approach.
\newblock In {\em Proceedings of 27th Annual ACM-SIAM Symposium on Discrete
  Algorithms (SODA)}. 712--729.

\bibitem[\protect\citeauthoryear{Tang, Wu, and Zhang}{Tang
  et~al\mbox{.}}{2020}]{tang2020towards}
{\sc Tang, Z.~G.}, {\sc Wu, X.}, {\sc and} {\sc Zhang, Y.} 2020.
\newblock Towards a better understanding of randomized greedy matching.
\newblock In {\em Proceedings of 52nd Annual ACM Symposium on Theory of
  Computing (STOC)}. 1097--1110.

\bibitem[\protect\citeauthoryear{Vondr{\'a}k}{Vondr{\'a}k}{2007}]{vondrak2007submodularity}
{\sc Vondr{\'a}k, J.} 2007.
\newblock Submodularity in combinatorial optimization.

\bibitem[\protect\citeauthoryear{Vondr{\'a}k}{Vondr{\'a}k}{2008}]{vondrak2008optimal}
{\sc Vondr{\'a}k, J.} 2008.
\newblock Optimal approximation for the submodular welfare problem in the value
  oracle model.
\newblock In {\em Proceedings of 40th Annual ACM Symposium on Theory of
  Computing (STOC)}. 67--74.

\bibitem[\protect\citeauthoryear{Vondr{\'a}k}{Vondr{\'a}k}{2013}]{vondrak2013symmetry}
{\sc Vondr{\'a}k, J.} 2013.
\newblock Symmetry and approximability of submodular maximization problems.
\newblock {\em SIAM Journal on Computing (SICOMP)\/}~{\em 42,\/}~1, 265--304.

\bibitem[\protect\citeauthoryear{Wajc}{Wajc}{2020}]{wajc2020rounding}
{\sc Wajc, D.} 2020.
\newblock Rounding dynamic matchings against an adaptive adversary.
\newblock In {\em Proceedings of 52nd Annual ACM Symposium on Theory of
  Computing (STOC)}. 194--207.

\bibitem[\protect\citeauthoryear{Wolsey}{Wolsey}{1982}]{wolsey1982analysis}
{\sc Wolsey, L.~A.} 1982.
\newblock An analysis of the greedy algorithm for the submodular set covering
  problem.
\newblock {\em Combinatorica\/}~{\em 2,\/}~4, 385--393.

\bibitem[\protect\citeauthoryear{Yao}{Yao}{1977}]{yao1977lemma}
{\sc Yao, A. C.-C.} 1977.
\newblock Probabilistic computations: Toward a unified measure of complexity.
\newblock In {\em Proceedings of 18th Symposium on Foundations of Computer
  Science (FOCS)}. 222--227.

\end{thebibliography}
